\documentclass[sts]{imsart}

\RequirePackage{amsthm,amsmath,amsfonts,amssymb}
\RequirePackage[numbers]{natbib}
\RequirePackage{graphicx}
\RequirePackage{multirow}
\RequirePackage{booktabs}
\RequirePackage{url}
\RequirePackage{todonotes}

\RequirePackage{pgfplots}
\usetikzlibrary{patterns}

\startlocaldefs
\theoremstyle{plain}
\newtheorem{lemma}{Lemma}[section]
\newtheorem{proposition}{Proposition}[section]

\theoremstyle{remark}
\newtheorem{decisionrule}{Decision Rule}

\newtheorem{remark}{Remark}[section]
\newtheorem{definition}{Definition}[section]

\newcommand{\btheta}{\boldsymbol{\theta}}
\newcommand{\bxi}{\boldsymbol{\xi}}
\newcommand{\bpsi}{\boldsymbol{\psi}}
\newcommand{\bkappa}{\boldsymbol{\kappa}}

\pgfmathdeclarefunction{gauss}{2}{%
  \pgfmathparse{1/(#2*sqrt(2*pi))*exp(-((x-#1)^2)/(2*#2^2))}%
}

\endlocaldefs

\begin{document}

\begin{frontmatter}
\title{Risk-aware product decisions in A/B tests with multiple metrics}
\runtitle{Combining multiple hypothesis tests into one decision}

\begin{aug}
\author[A]{\fnms{M\aa rten}~\snm{Schultzberg}\ead[label=e1]{mschultzberg@spotify.com}},
\author[A]{\fnms{Sebastian}~\snm{Ankargren}\ead[label=e2]{sebastiana@spotify.com}}
\and
\author[A]{\fnms{Mattias}~\snm{Fr\aa nberg}\ead[label=e3]{mfranberg@spotify.com}}
\address[A]{Experimentation Platform team, Spotify,
Stockholm, Sweden \printead[presep={\ }]{e1}.}
\end{aug}

\begin{abstract}
In the past decade, A/B tests have become the standard method for making product decisions in tech companies. They offer a scientific approach to product development, using statistical hypothesis testing to control the risks of incorrect decisions. Typically, multiple metrics are used in A/B tests to serve different purposes, such as establishing evidence of success, guarding against regressions, or verifying test validity. 
To mitigate risks in A/B tests with multiple outcomes, it's crucial to adapt the design and analysis to the varied roles of these outcomes. This paper introduces the theoretical framework for decision rules guiding the evaluation of experiments at Spotify. First, we show that if guardrail metrics with non-inferiority tests are used, the significance level does not need to be multiplicity-adjusted for those tests. Second, if the decision rule includes non-inferiority tests, deterioration tests, or tests for quality, the type II error rate must be corrected to guarantee the desired power level for the decision. We propose a decision rule encompassing success, guardrail, deterioration, and quality metrics, employing diverse tests. This is accompanied by a design and analysis plan that mitigates risks across any data-generating process.The theoretical results are demonstrated using Monte Carlo simulations.
\end{abstract}


\end{frontmatter}
\section{Introduction}
Randomized experiments are the gold standard for providing evidence for causal relationships. Modern technology companies use A/B tests, a randomized controlled trial in a digital setting, extensively to evaluate the efficacy of new changes to their products. These products include ride-sharing apps, search engines, streaming services, recommendations, and more. Ultimately, the goal of these experiments is to decide whether or not to release a product change more widely.

Most of the literature on statistical inference for randomized experiments focuses on a single hypothesis test of a single outcome, and how to bound the type I and type II error rates for that test. However, experiments are not univariate tests of isolated outcomes. Instead, the risks that matter are the risks of making the incorrect decision for the product. For example, at a tech company like Spotify, we want to limit how often we release product changes that show an improvement when there truly is none, and how often we refrain from releasing changes that lead to improvements but we fail to find. These types of decisions typically include results from several hypothesis tests. 
Experiments usually involve results for multiple outcomes, and making a single decision based on these multiple outcomes can be challenging. For example, some of the outcomes, what we will refer to as `metrics', may show improvements, while others show none or even negative effects.

In the online experimentation literature, the only aspect of multi-test decision making that is extensively covered is multiple-testing correction. Multiple-testing corrections, such as Bonferroni, Holm \citep{holm1979} and Hommel \citep{hommel1988}, bound the type I error rate of an implied decision rule that declares what decision you will make based on the results of the individual hypothesis tests. As we will discuss extensively in this paper, unless your desired decision rule matches the rule implied by the multiple-testing correction, it is typically incorrect.   

In this paper, we show how it is possible to formalize the decision-making process of experiments without leaving the standard hypothesis testing framework. The key to ensuring that you obtain the intended risk bounds for the product decision is to explicitly specify a \textit{decision rule}. A decision rule exhaustively specifies what product decision you will make based on the results of your experiment. Importantly, to bound the risks of making an erroneous decision, the design and analysis of your experiment must be closely aligned with the decision rule.

Articulating the decision rule is important for several reasons. Being unclear about what outcomes lead to a positive product decision means that there is no mechanism for properly controlling the risks of the experiment at the level that matters to the company, namely the decision to ship the feature or not. Additionally, a lack of an articulated and standardized decision rule can mean that different teams or parts of the organization hold themselves to different standards. Our decision rule framework is a simple but effective approach for combating those issues.

The decision rule framework helps standardize the analysis of experiments and is a useful tool for experimentation platforms. What the decision rule includes can be made more or less flexible. For example, new experiments can be forced to demonstrate that important company metrics are not negatively impacted while selecting the set of metrics that should show an improvement is made completely up to the experimenter. Even if the choice of metrics is completely arbitrary with no metrics made mandatory by the platform, the decision rule approach promotes a shared understanding of what a successful experiment is.

Throughout this paper, and without loss of generality, we only consider experiments with two groups to simplify notation.
In addition, we only consider one-sided tests, although more than one one-sided test might be applied to each metric. We limit ourselves to one-sided tests as there must be an intended direction for a change in the metric to map to a measurable improvement in the product. For simplicity, we assume that all metrics improve when they increase. Moreover, we assume that each statistical hypothesis test is valid and achieves its type I and type II error rates exactly if the experiment is designed accordingly. 

\subsection{Related literature}
Decision theory is a formal mathematical framework for formalizing decision problems under uncertainty, see e.g. \cite{berger2013decisiontheory} for an introduction. Although this theory is comprehensive and flexible, it is non-trivial for most people, and it moves the decision problem far from the (to most experimenters) familiar hypothesis-testing realm. Since modern tech companies are often having many teams experimenting independently, it is not plausible to have decision theory experts available to all. 

Another popular alternative for decision making in A/B tests where the goal is to evaluate using several outcomes is to use a so-called overall evaluation criterion (OEC), see e.g. \cite{liu2022OEC} for a recent introduction. An OEC removes the problem with several possibly contradictory results by using just a single metric. This metric can either be a metric that serves as a proxy for all the necessary aspects of the important outcomes, or it is a function, like a linear combination, of a selected set of metrics. Designing an appropriate OEC generally requires prolonged research and strong alignment within the organization. For larger companies, a single OEC may not even suffice because the business itself is diverse. Moreover, if the OEC is a complicated function of several outcomes, it can be difficult for experimenters to understand their results. Even when an OEC is used, it is not common that this metrics include all quality tests and metrics that should not deteriorate. That is, an OEC is typically a metric that trades off various outcomes to define success, but as we will show in the paper, a product decision rule can also explicitly include the efforts to avoid end-user harm and experiment invalidity which in turn affects the power analysis and design of the experiment.  

In the clinical trial literature (for an overview, see ~\citealp[ch. 4]{dmitrienko2009};~\citealp{Offen2007}), so-called multiple endpoint experiments, are experiments with more than one outcome (metric). The endpoints can be both efficacy (success metrics) and safety endpoints (guardrail metrics). In some trials, there are also primary endpoints and secondary endpoints, where the secondary endpoints are only evaluated if primary endpoints are significantly changed. Clearly, this setting closely resembles the online experimentation setting. Similarly to deciding whether a new drug or treatment is safe and effective, deciding whether to ship a new feature is a composite decision that involves a potentially complex interplay between all endpoints. Various experimental design and analysis methods have been applied in the clinical trial setting \citep{neuhauser2006}: hierarchical testing where primary endpoints are tested before secondary \citep{dmitrienko2003}, global assessment measures where the endpoints are aggregated within each patient first then analyzed with standard statistical methods \citep{obrien1984}, closed testing where a global null is tested first before proceeding to more specific hypotheses \citep{lehmacher1991}. However, at companies like Spotify, there's a strong need to standardize the design, analysis and decision process of multiple endpoint experiments to allow non-statisticians to evaluate product changes in a safe and efficient way.


\section{Types of metrics and their hypotheses}
 Throughout the paper, we assume that the aim of the statistical analysis of an experiment is to make a binary decision regarding the success of the treatment. This decision can be e.g. whether or not to go ahead with the next step of validation for a medical treatment, or whether or not to ship a product change. Our goal is to bound the type I and type II error for this decision in repeated experimentation. We will focus on the product decision example, but the results are applicable in a wider, more general setting.

\subsection{Types of metrics}
In modern online experimentation, experiments are evaluated using multiple metrics. Based on the results from each metric, typically estimates of the average treatment effects, the experimenters make a decision whether to ship the feature more widely. The heuristics underlying the decision-making process are seldom transparent. In large organizations, these decision processes often vary substantially from team to team. At Spotify, we have introduced a standardized way of providing a recommendation for a suggested course of action given the outcomes for a set of metrics that belong to different categories. We call this recommendation a shipping recommendation. These recommendations are powered by a decision rule that includes four types of metrics and their associated hypotheses and tests, given by
\begin{enumerate}
    \item \textbf{Success metrics}. Metrics that we aim to improve, tested with superiority tests.
    \item \textbf{Guardrail metrics}. Metrics that we do not want to see deteriorate more than a certain threshold, tested with non-inferiority tests.
    \item \textbf{Deterioration metrics}. Metrics that should not deteriorate, tested with inferiority tests. 
    \item \textbf{Quality metrics}. Metrics that verify the integrity and validity of the experiment itself.
\end{enumerate}
A metric can belong to multiple categories. For example, at Spotify, all success and guardrail metrics also belong to deterioration metrics. We will elaborate on the details of this in later sections, but the implication is that we monitor all metrics for regressions, even if our goal and hypothesis is for them to improve. Quality metrics are not always metrics in the traditional sense. For example, a crucial experiment-quality metric is the sample ratio mismatch test \citep{Fabijan2018}. This is not a metric in the traditional sense, but rather a goodness of fit test of proportions. 
\begin{table}[hbt!]
    \centering
    \begin{tabular}{rll}
    \toprule
    & & Direction used \\
       Metric type  & Metric name  &  in test \\
         \midrule
    Success & Music minutes played & Increase \\
    Guardrail & Podcast minutes played & Increase  \\
    \midrule
    Deterioration & Music minutes played & Decrease \\
    Deterioration & Podcast minutes played & Decrease  \\
    Deterioration & Share of users with a crash & Increase \\
    \midrule
    Quality & Sample ratio mismatch & N/A\\
        \bottomrule
    \end{tabular}
    \caption{Example set of metrics for an experiment. The success and guardrail metrics in the experiment appear a second time as deterioration metrics to ensure they are not unknowingly moving in the wrong direction. }
    \label{tab:metric_example}
\end{table}
Table \ref{tab:metric_example} shows an example of a set of metrics used in an experiment. In this case, the experiment attempts to increase the minutes played of music, but includes podcast minutes played to verify that overall consumption does not increase at the expense of podcast consumption. Both metrics are also included as deterioration metrics to ensure that they are not moving in the direction that is opposite from what we would expect and hope. Additionally, the share of users that experience a crash is included as deterioration metric. The only quality metric in this case is a sample ratio mismatch metric.

In the following, we use "inferiority test" to equivalently mean "deterioration test", which, given our standarization that a decrease is a regression, is a test for the deterioration of a metric. With some abuse of language, we sometimes say "metric $x$ was significantly superior" to mean that the "treatment was significantly superior to control with respect to metric $x$".

\subsection{Hypotheses for different types of metrics}
The different categories of metrics serve different purposes, which by extension means that their associated statistical hypotheses are different. Table \ref{tab:hypothesis_general} displays the hypotheses for these categories of metrics used in the decision rules considered in this paper.
\begin{table*}[hbt]
    \centering
    \begin{tabular}{r|llll}
    \toprule
    & && Effect used for& \\
       Metric  & Null  & Alternative &  power analysis & Status quo\\
         \midrule
        Success & $\delta \leq 0$  & $\delta > 0$ & $\delta = MDE$&  $\delta = 0$\\
        Guardrail & $\delta \leq -NIM$ & $\delta > -NIM$ &$\delta = 0$&$\delta = 0$\\  
        Deterioration & $\delta \geq 0$ & $\delta < 0$ & N/A  & $\delta = 0$\\
        \bottomrule
    \end{tabular}
    \caption{Hypotheses for the the three main types of metrics considered in the decision rules, where $\delta$ is the estimand, like the average treatment effect, of interest. The hypotheses of the quality tests are left out because they typically are not using the same kind of estimands as the others with varying hypotheses as a consequence. The minimum detectable effect (MDE) is the effect size used for success metrics when designing the experiment. The non-inferiority margin (NIM) is the tolerance level of regression used in the non-inferiority tests used for guardrail metrics. Status quo is not a hypothesis in the traditional sense, but a scenario of interest later in the paper. }
    \label{tab:hypothesis_general}
\end{table*}
While the hypotheses are similar–––and to some degree even opposites of one another–––they give rise to distinct interpretations. For example, the alternative of the non-inferiority test for which we design the experiment to be powered is the null hypothesis for the superiority test, which means that under the null hypothesis a guardrail metric has deteriorated by $NIM$ (non-inferiority margin). We also consider a third hypothesis-like scenario, the "status quo" in which no metric has moved, to facilitate our discussion. The status quo scenario is thus under the alternative for the non-inferiority and under the null hypothesis for the superiority tests and deterioration tests. 

\section{Type I and Type II error rates for decision rules including superiority and non-inferiority tests }
The categorization of metrics into success, guardrail, deterioration and quality metrics naturally paves the way for a decision rule that combines the results of metrics in each category appropriately. The decision rule, in turn, implies various multiple-testing corrections for the statistical inference to control both the type I and type II error rates of the shipping decision as intended. In this section, we start by establishing some fundamental results for superiority and non-inferiority testing, and, more generally, for union-intersection and intersection-union testing that these rely on. These results are then used to construct a decision rule that includes superiority and non-inferiority tests. In subsequent sections, we include deterioration tests and finally quality metrics. 

The goal of the designs in this paper is to bound the family-wise error rates of the decision. In the following, we focus our discussion on Bonferroni-based adjustments for multiple comparisons that let us use results of individual tests to evaluate the joint global hypothesis. There are two main contributing factors to why we choose this approach: 
\begin{enumerate}
    \item Interpretability is greatly simplified when experimenters can view individual results for metrics (including confidence intervals).
    \item By evaluating a global hypothesis consisting of individual of hypotheses for multiple metrics through individual tests, we can fit our framework into a large, scalable experimentation platform where a decision rule approach fits seamlessly into other core experimentation functionality like sequential testing, variance reduction, and more.
\end{enumerate}

While other multiple-testing adjustment procedures often are more powerful, they generally achieve more power at the expense of interpretability and ease of understanding. For example, Holm's multiple correction method \citep{holm1979} is generally more powerful than Bonferroni's, but the associated confidence intervals under are more complicated and do not always yield finite bounds \citep{guilbaud2014sharper}. Even when confidence intervals are available, explaining the correction and how it affects the intervals to experimenters is non-trivial.
\subsection{The composite hypotheses of superiority and non-inferiority tests}
To facilitate our discussion about overall decision rules, we first describe the hypotheses for the superiority and non-inferiority tests used for success and guardrail metrics.
Let $H_0^{(\mathcal{L})}$ be the null hypothesis for a set of tests, where the superscript indicates the type of the tests. We occasionally drop the superscript when it is given by the context. Let $H_\mathcal{A}^{(\mathcal{L})}$ be the alternative hypothesis.
\begin{definition}[At-least-one testing for superiority]
A composite hypothesis for superiority that is rejected if at least one of the $S$ subhypotheses is rejected is described by
\begin{align*}
    H_0^{(\mathcal{S})}: \bigcap_{i=1}^S H_{0, i}^{(\mathcal{S})} \text{ versus } H_{\mathcal{A}}^{(\mathcal{S})}: \bigcup_{i=1}^S H_{\mathcal{A}, i}^{(\mathcal{S})}
\end{align*}
where 
\begin{align*}
    H_{0, i}^{(\mathcal{S})}: \delta_i \leq 0 \text{ versus } H_{\mathcal{A}, i}^{(\mathcal{S})}: \delta_i > 0
\end{align*}\label{def:sup}
\end{definition}
Definition \ref{def:sup} says that we reject the global null hypothesis if \textit{any} of the subhypotheses are rejected. This is the decision rule that most multiple testing corrections bound the type I error rate for, including the Bonferroni and Sidak corrections. In other words, if any metric improves significantly, we ship the product change.

\begin{definition}[All-or-none testing for non-inferiority]
A composite hypothesis for non-inferiority that is rejected if all the $G$ subhypotheses is rejected is described by
\begin{align*}
    H_{0}^{(\mathcal{G})}: \bigcup_{j=1}^G H_{0,j}^{(\mathcal{G})} \text{ versus } H_{\mathcal{A}}^{(\mathcal{G})}: \bigcap_{j=1}^G H_{\mathcal{A}, j}^{(\mathcal{G})}.
\end{align*}
where 
\begin{align*}
    H_{0, j}^{(\mathcal{G})}: \delta_j \leq -\epsilon_j \text{ versus } H_{\mathcal{A}, j}^{(\mathcal{G})}: \delta_j > -\epsilon_j.
\end{align*}
\label{def:noninf}
\end{definition}
Definition \ref{def:noninf} says that we will reject the global null hypothesis only if \textit{all} subhypotheses are rejected. This kind of global hypothesis, although it is arguably natural for non-inferiority tests, has not received much attention in the online experimentation literature.  

The two testing procedures outlined in Definition \ref{def:sup} and \ref{def:noninf} are both well-studied in statistics, where they are known as union-intersection (UI) and intersection-union (IU) testing, respectively \citep{dmitrienko2009}. 

\subsection{Bounding the type I and type II error rates for UI and IU testing}

Because the at-least-one testing for superiority is based on the UI testing principle, individual tests must be carried out at an adjusted significance level to ensure that the family-wise error rate is bounded by the nominal significance level. On the other hand, no adjustment is necessary to the nominal power level used when designing the experiment. The at-least-one testing principle implies that the composite test has a family-wise type II error rate that is bounded from above by the nominal level. For the all-or-none testing used to establish non-inferiority, it is not necessary to multiplicity adjust the individual significance level used in the tests. 

The error bounding of the decision rules that we propose in this paper heavily relies on the bounding of error rates for UI and IU testing. Therefore, we establish this more formally below. Let $\theta$ be the parameter of interest, and let $\theta_0$ be the parameter value under $H_0$ and let $\theta_\mathcal{A}$ be the parameter value under $H_\mathcal{A}$ for which the power is $1-\beta$.

For tests that rely on the UI principle, only the individual type I error rate needs to be corrected to bound the overall type I and type II error rates. We formalize this in Lemma \ref{lemma:sig_pow_union_intersec}. 
\begin{lemma}[Significance and power level adjustments for union-intersection testing.]\label{lemma:sig_pow_union_intersec}
Consider a UI testing problem of $M$ hypotheses on the form
\begin{align*}
    H_0: \bigcap_{i=1}^M H_{0, i} \text{ versus } H_{\mathcal{A}}: \bigcup_{i=1}^M H_{\mathcal{A}, i}
\end{align*}
with target nominal type I and type II error rates $\alpha$ and $\beta$. Let $\btheta \in \mathbb{R}^M$ be the $M$-dimensional parameter being tested, and let $\Theta_0$ and $\Theta_{\mathcal{A}}$ be the parameter spaces corresponding to the null and alternative hypotheses, respectively, such that the rejection regions of the tests satisfy
\begin{align*}
    \mathrm{Pr}(H_{0,i} \text{ is rejected} \mid \btheta \in \Theta_0)&=\frac{\alpha}{M}\\    
    \mathrm{Pr}(H_{0,i} \text{ is rejected} \mid \btheta \in \Theta_{\mathcal{A}})&=1-\beta.
\end{align*}
Then, the overall test of $H_0$ satisfies
\begin{align*}
    \mathrm{Pr}(H_0 \text{ is rejected} \mid  \btheta \in \Theta_0) &\leq \alpha \\ 
    \mathrm{Pr}(H_0 \text{ is rejected} \mid \btheta \in \Theta_{\mathcal{A}})&\geq 1-\beta.
\end{align*}
\label{lemma:UI}
\end{lemma}
\begin{proof}
We start by proving that $\mathrm{Pr}(H_0 \text{ is rejected} \mid \btheta \in \Theta_0) \leq \alpha $.
Let $R$ be the event that $H_0$ is rejected, and $R_i$ be the event that $H_{0,i}$ is rejected. We have:
\begin{align*}
\mathrm{Pr}(R \mid \btheta \in \Theta_0) &= \mathrm{Pr}\left(\bigcup_{i=1}^M R_i \mid \btheta \in \Theta_0\right)\\ &\leq \sum_{i=1}^M \mathrm{Pr}(R_i \mid \btheta \in \Theta_0).
\end{align*}
Now, if each $H_{0,i}$ is tested at the significance level $\alpha/M$, then:
\begin{align*}
\mathrm{Pr}(R_i \mid \btheta \in \Theta_0) \leq \frac{\alpha}{M} \quad \text{for } i = 1, 2, \ldots, M
\end{align*}
Substituting this into the previous equation:
\begin{align*}
\mathrm{Pr}(R \mid \btheta \in \Theta_0) \leq \sum_{i=1}^M \frac{\alpha}{M} = \alpha,
\end{align*}
which shows that the probability of rejecting the null hypothesis when it is true is less than or equal to $\alpha$.

Next, we prove that $\mathrm{Pr}(R \mid \btheta \in \Theta_{\mathcal{A}})\geq 1-\beta$.
We now have:
\begin{align*}
\mathrm{Pr}(R \mid \btheta \in \Theta_{\mathcal{A}}) &= \mathrm{Pr}\left(\bigcup_{i=1}^M R_i \mid \btheta \in \Theta_{\mathcal{A}}\right) \\
&= 1 - \mathrm{Pr}\left(\bigcap_{i=1}^M \lnot R_i \mid \btheta \in \Theta_{\mathcal{A}}\right) \\
&\geq 1 - \min_i\mathrm{Pr}(\lnot R_i \mid \btheta \in \Theta_{\mathcal{A}})
\end{align*}
where the the last step is by the Fréchet inequality.
Now, if each test is designed at the power level $1-\beta$, then:
\begin{align*}
\min_i \mathrm{Pr}(\lnot R_i \mid \btheta \in \Theta_{\mathcal{A}}) \leq \beta \quad \text{for } i = 1, 2, \ldots, M
\end{align*}
Substituting this into the previous equation:
\begin{align*}
\mathrm{Pr}(R \mid \btheta \in \Theta_{\mathcal{A}}) \geq 1 - \beta,
\end{align*}
which finishes the proof.
\end{proof}

Contrary to the UI case, only the type II error rate must be adjusted for tests that rely on the IU principle. We formalize this in Lemma \ref{lemma:sig_pow_intersec_union}.
\begin{lemma}[Significance and power level adjustments for intersection-union testing.]\label{lemma:sig_pow_intersec_union}
Consider an intersection union testing problem of $M$ hypotheses on the form
\begin{align*}
    H_0: \bigcup_{i=1}^M H_{0, i} \text{ versus } H_{\mathcal{A}}: \bigcap_{i=1}^M H_{\mathcal{A}, i}
\end{align*}
with target nominal type I and type II error rates $\alpha$ and $\beta$. Let $\btheta\in \mathbb{R}^M$ be the $M$-dimensional parameter being tested, and let $\Theta_0$ and $\Theta_{\mathcal{A}}$ be the parameters spaces under the null and alternative hypotheses, respectively, such that the rejection regions of the tests satisfy
\begin{align*}
    \mathrm{Pr}(H_{0,i} \text{ is rejected} \mid \btheta \in \Theta_0)&=\alpha\\    
    \mathrm{Pr}(H_{0,i} \text{ is rejected} \mid \btheta \in \Theta_{\mathcal{A},i})&=1-\frac{\beta}{M}.
\end{align*}
Then, the overall test of $H_0$ satisfies
\begin{align*}
    \mathrm{Pr}(H_0 \text{ is rejected} \mid \btheta \in \Theta_0) &\leq \alpha \\
    \mathrm{Pr}(H_0 \text{ is rejected} \mid \btheta \in \Theta_{\mathcal{A}})&\geq 1-\beta.
\end{align*}
\label{lemma:IU}
\end{lemma}
\begin{proof}
We  again let $R$ be the event that $H_0$ is rejected, and $R_i$ be the event that $H_{0,i}$ is rejected and start by proving that $\mathrm{Pr}(R \mid H_0) \leq \alpha$. We then have that
\begin{align*}
  \mathrm{Pr}(R \mid \btheta \in \Theta_0) &=  \mathrm{Pr}\left(\bigcap_{i=1}^M R_i \mid \btheta \in \Theta_0\right)\\
  &\leq \min_{i} \mathrm{Pr}(R_i \mid \btheta \in \Theta_0) = \alpha 
\end{align*}
which proves the first part. 
For proving that $\mathrm{Pr}(R \mid \btheta \in \Theta_{\mathcal{A}}) \geq 1 - \beta$, we have that
\begin{align*}
    \mathrm{Pr}(R \mid \btheta \in \Theta_{\mathcal{A}}) &= \mathrm{Pr}\left(\bigcap_{i=1}^M R_i \mid \btheta \in \Theta_{\mathcal{A}} \right)\\
     &= 1- \mathrm{Pr}\left(\bigcup_{i=1}^M \lnot R_i \mid \btheta \in \Theta_{\mathcal{A}} \right)\\
     &\geq 1- \sum_{i=1}^M \mathrm{Pr}(\lnot R_i \mid \btheta \in \Theta_{\mathcal{A}})\\
     &\geq 1- \sum_{i=1}^M  \frac{\beta}{M} = 1- \beta\\
\end{align*}
which finishes the proof.  
\end{proof}
The results of the lemmas can be intuitively understood by considering in what situations multiple chances arise. In the UI testing setting of Lemma \ref{lemma:UI}, it suffices that at least one individual null hypothesis is rejected for the overall hypothesis to be rejected. Since each individual hypothesis $H_i$ constitutes one chance of rejecting the overall hypothesis, we must correct for these multiple possibilities. On the other hand, in the IU testing case of Lemma \ref{lemma:IU}, all individual null hypotheses must be rejected for the overall hypothesis to be rejected. Equivalently stated, there are multiple chances of \emph{not} rejecting---and this is what we must correct for. Therefore, to ensure that the power of the IU test is bounded from below by the desired nominal level, the type II error rate $\beta$ used in each individual test is multiplicity adjusted using a standard Bonferroni correction. 

\subsection{Bounding the error rates for a decision rule including both success and guardrail metrics}
The decision rule we use combines superiority and non-inferiority testing, and appropriate adjustments to individual-level tests follow immediately from the results of the previous section. Before we discuss these adjustments, we first describe the decision rule. 

\begin{decisionrule}\label{dr1}
Ship the change if and only if:
\begin{enumerate}
    \item at least one success metric is significantly superior
    \item all guardrail metrics are significantly non-inferior.
\end{enumerate}
That is, ship if and only if the global superiority/non-inferiority null hypothesis is rejected in favor of the alternative hypothesis.
\end{decisionrule}

We can now establish how to adjust Decision Rule \ref{dr1} to ensure that family-wise type I or type II error rates are not inflated. The decision rule consists of two levels of testing and can be mathematically stated as:
\begin{align*}
    H_0^{(\mathcal{S})}\bigcup H_0^{(\mathcal{G})} \text{ versus } H_{\mathcal{A}}^{(\mathcal{S})}\bigcap H_{\mathcal{A}}^{(\mathcal{G})}.
\end{align*}
The testing problem is thus an intersection-union test consisting of $G+1$ hypotheses---$G$ hypotheses for the guardrail metrics, and one additional composite hypothesis concerning the at-least-one test for superiority. By applying Lemma \ref{lemma:UI} and \ref{lemma:IU}, we can establish the necessary adjustments for the decision rule, formalized in Proposition \ref{prop:alpha_beta_under_correction}.
\begin{proposition}\label{prop:alpha_beta_under_correction}
    Let $\alpha$ and $\beta$ be the nominal type I and type II rates that the decision rule should satisfy, and let there be $G$ guardrail metrics, and $S$ success metrics, where $S>0$. If 
    \begin{enumerate}
        \item each guardrail metric is tested using the significance level $\alpha$ and designed for the power level $1-\dfrac{\beta}{G+1}$
        \item each success metric is tested using the significance level $\frac{\alpha}{S}$ and designed for the power level $1-\dfrac{\beta}{G+1}$,
    \end{enumerate}
    then the error rates for Decision Rule \ref{dr1} are at most $\alpha$ and $\beta$.
\end{proposition}
\begin{proof}
Let now $R_{(\mathcal{S})}$ and $R_{(\mathcal{G})}$ denote the events that, respectively, the composite success and guardrail hypotheses are rejected. Furthermore, let $\Theta_0^{(\mathcal{L})}$ be the parameter space under the null hypothesis such that each individual tests are of level at most $\alpha$ for the class of tests $\mathcal{L}$, and let $$\Theta_{\beta, \mathcal{A}}^{(\mathcal{L})} \subseteq \Theta_{\mathcal{A}}^{(\mathcal{L})}$$ be a part of the parameter space for which all individual tests achieve a power of at least $1-\beta$.
We start by proving that $$\mathrm{Pr}\left(R_{(\mathcal{S})}\bigcap R_{(\mathcal{G})} \mid \btheta^{(\mathcal{S})} \in \Theta_0^{(\mathcal{S})}, \btheta^{(\mathcal{G})} \in \Theta_0^{(\mathcal{G})}\right)\leq \alpha $$. We have that
\begin{align*}
&\mathrm{Pr}\left(R_{(\mathcal{S})}\bigcap R_{(\mathcal{G})} \mid \btheta^{(\mathcal{S})} \in \Theta_0^{(\mathcal{S})}, \btheta^{(\mathcal{G})} \in \Theta_0^{(\mathcal{G})}\right) \\&\leq  \min \Big( 
\mathrm{Pr}\left(R_{(\mathcal{S})} \mid \btheta^{(\mathcal{S})} \in \Theta_0^{(\mathcal{S})}\right),\\
&\mathrm{Pr}\left(R_{(\mathcal{G})} \mid \btheta^{(\mathcal{G})} \in \Theta_0^{(\mathcal{G})}\right)
\Big).
\end{align*}
By Lemmas \ref{lemma:sig_pow_intersec_union} and \ref{lemma:sig_pow_union_intersec}, and the significance level used by each test given the conditions of the proposition, we know that the upper bound for both these probabilities is $\alpha$, which means that 
\begin{align*}
  &\min \Big( 
\mathrm{Pr}\left(R_{(\mathcal{S})} \mid \btheta^{(\mathcal{S})} \in \Theta_0^{(\mathcal{S})}\right),\\ &
\mathrm{Pr}\left(R_{(\mathcal{G})} \mid \btheta^{(\mathcal{G})} \in \Theta_0^{(\mathcal{G})}\right)
\Big)\\ &\leq \min (\alpha, \alpha) = \alpha.
\end{align*}
Second, we prove that 
$$\mathrm{Pr}\left(R_{(\mathcal{S})}\bigcap R_{(\mathcal{G})} \mid \btheta^{(\mathcal{S})} \in \Theta_{\beta, \mathcal{A}}^{(\mathcal{S})}, \btheta^{(\mathcal{G})} \in \Theta_{\beta, \mathcal{A}}^{(\mathcal{G})}\right)\geq 1-\beta.$$ We have that
\begin{align*}
    &\mathrm{Pr}\left(R_{(\mathcal{S})}\bigcap R_{(\mathcal{G})} \mid \btheta^{(\mathcal{S})} \in \Theta_{\beta, \mathcal{A}}^{(\mathcal{S})}, \btheta^{(\mathcal{G})} \in \Theta_{\beta, \mathcal{A}}^{(\mathcal{G})} \right) \\&= 1 - \mathrm{Pr}\left( \lnot R_{(\mathcal{S})} \bigcup \lnot R_{(\mathcal{G})} \mid \btheta^{(\mathcal{S})} \in \Theta_{\beta, \mathcal{A}}^{(\mathcal{S})}, \btheta^{(\mathcal{G})} \in \Theta_{\beta, \mathcal{A}}^{(\mathcal{G})} \right)\\
    &\geq 1- \Big( \mathrm{Pr}(\lnot R_{(\mathcal{S})} \mid \btheta^{(\mathcal{S})} \in \Theta_{\beta, \mathcal{A}}^{(\mathcal{S})}) +\\& \mathrm{Pr}( \lnot R_{(\mathcal{G})} \mid \btheta^{(\mathcal{G})} \in \Theta_{\beta, \mathcal{A}}^{(\mathcal{G})} ) \Big), 
\end{align*}
where the inequality follows by the union bound. By Lemmas \ref{lemma:sig_pow_union_intersec} and \ref{lemma:sig_pow_intersec_union} we know that since $\beta/(G+1)$ is used for all success and guardrail metrics, it follows that
$\mathrm{Pr}(\lnot R_{(\mathcal{S})} \mid \btheta^{(\mathcal{S})} \in \Theta_{\beta, \mathcal{A}}^{(\mathcal{S})})\leq \frac{\beta}{G+1}$ and $\mathrm{Pr}(R_{(\lnot \mathcal{G})} \mid \btheta^{(\mathcal{G})} \in \Theta_{\beta, \mathcal{A}}^{(\mathcal{G})} ) \leq \frac{G\beta}{G+1}$. Plugging this into the previous inequality gives
\begin{align*}
    &\mathrm{Pr}\left(R_{(\mathcal{S})}\bigcap R_{(\mathcal{G})} \mid \btheta^{(\mathcal{S})} \in \Theta_{\beta, \mathcal{A}}^{(\mathcal{S})}, \btheta^{(\mathcal{G})} \in \Theta_{\beta, \mathcal{A}}^{(\mathcal{G})}\right) \\&\geq 1 - \left(\frac{\beta}{G+1}  + \frac{G\beta}{G+1} \right) = 1-\beta  
\end{align*}
which concludes the proof.
\end{proof}
Proposition \ref{prop:alpha_beta_under_correction} assumes the existence of at least one success metric. If there is none, Lemma \ref{lemma:sig_pow_intersec_union} can be used directly to see that the same result holds, but with individual power levels $1-\beta/G$.

\subsection{Power corrections for non-inferiority testing}

Power corrections (or, equivalently, "beta corrections") have, to the best of our knowledge, not been discussed in the online experimentation literature. However, the importance of adjusting power in the presence of IU testing is known in the statistical literature (\citealp{Kong2004};~\citealp{Offen2007};~\citealp[ch. 4]{dmitrienko2009}). These corrections are crucial to ensure appropriate levels of power when using non-inferiority tests, since the goal is to find evidence for non-inferiority in all tests simultaneously. The intuition for power corrections for multiple guardrails tested with non-inferiority tests is analogous to standard multiple-testing corrections of type I error rates. If, e.g., all guardrail metrics are independent and each guardrail metric has a NIM for which the power is 80\%, then the probability that at least one of them is not significantly non-inferior under the alternative quickly goes towards one as the number of guardrail metrics increase. Under independence, the number of significant guardrail metrics out of a total of $G$ guardrail metrics is distributed as $\mathrm{Bin}(G,0.8)$. The power for the decision rule, assuming only independent guardrail metrics, is simply $0.8^G$, which decays quickly in $G$. 
Figure \ref{fig:power_corr_illustration} displays how the power of the decision rule drops. Already at 10 guardrail metrics, the simultaneous power is less than 11\% without adjustment.
\begin{figure}\centering
\includegraphics[width=\linewidth]{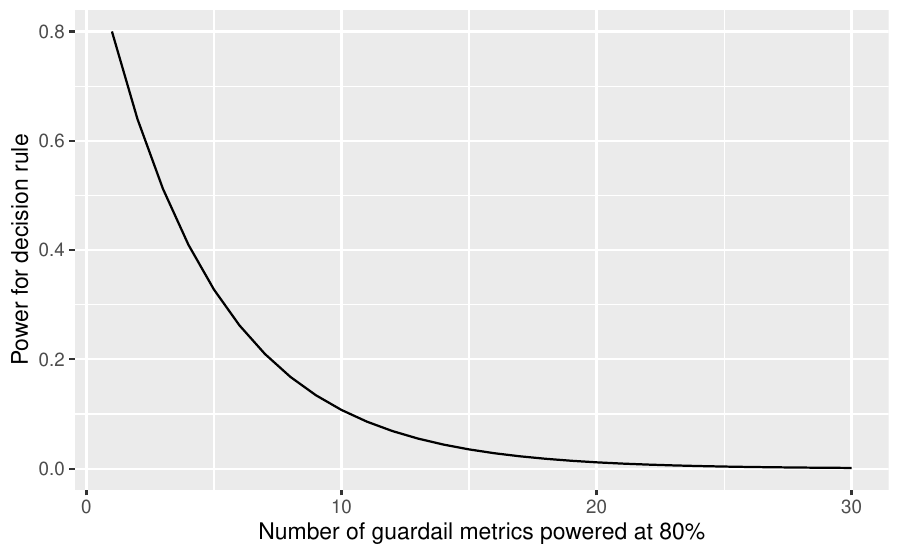}
\caption{Power for simultaneous rejection of all non-inferiority null hypotheses when the number of guardrail metrics increase and the power level is not adjusted appropriately. All metrics are independent, and individually powered to 80 percent.}\label{fig:power_corr_illustration}
\end{figure}
For theoretical bounds on the type I and type II error rates for Decision rule \ref{dr1} under perfect linear correlation and independence, respectively, see appendix Section \ref{sec:app:examples_fpr_power_dr1}.

\section{Extending the decision rule with deterioration and quality metrics}
Deterioration tests are inferiority tests that aim to capture regressions in metrics. They can be applied to metrics already present as a success or a guardrail metric, or to metrics that are only included in this category as deterioration metrics. The test checks if the metric is significantly deteriorating. That is, if the treatment group is significantly inferior to the control group with respect to the metric of interest.
Deterioration tests for success and guardrail metrics attempt to identify significant regressions, which would, if they exist, speak against the success of the experiment. Neither the superiority test used for success metrics or the non-inferiority test used for guardrail metrics would on its own indicate a regression. Knowing when regressions occur is essential when running experiments.

In practice, in addition to success, guardrail and detrioriation metrics, also experiment-quality metrics are used. For example, at Spotify, we include a set of tests and metrics that evaluate the quality of the experiment. These include a test for balanced traffic through a sample ratio mismatch test, and a test for pre-exposure bias. By including deterioration and quality tests in the decision rule, the complexity to manage the risks of an incorrect decision increases.  
Decision Rule \ref{dr2} formalizes the complete decision rule that is used at Spotify. 
\begin{decisionrule}\label{dr2}
Ship the change if and only if:
\begin{enumerate}
    \item at least one success metric is significantly superior
    \item all guardrail metrics are significantly non-inferior
    \item none of the success, guardrail or deterioration metrics are significantly inferior
    \item none of the quality tests are significantly rejecting quality
\end{enumerate}
That is, ship if and only if the global superiority/non-inferiority null hypothesis is rejected in favor of the alternative hypothesis, the inferiority null hypothesis is not rejected for any metric, and no quality test is significant. 
\end{decisionrule}
Proposition \ref{prop:alpha_beta_under_correction_with_deterioration_and_quality} displays the corresponding correction, to bound the error rates of Decision Rule \ref{dr2}. 
\begin{proposition}\label{prop:alpha_beta_under_correction_with_deterioration_and_quality}
    Let $\alpha$ and $\beta$ be the nominal type I and type II rates that the decision rule should satisfy, and let there be $G$ guardrail metrics, $S$ success metrics, $D$ deterioration metrics in addition to the success and guardrail metrics, and $Q$ be quality metrics/tests. Let $\alpha_-$ be the total nominal type I rate for the deterioration and quality tests where $\alpha_-<\beta$. If   
    \begin{enumerate}
        \item each deterioration and quality test is performed with significance level $\alpha_-^*= \frac{\alpha_-}{S+G+D+Q}$
        \item each success metric is tested using a superiority test using the significance level $\alpha^*=\frac{\alpha}{S}$
        \item each guardrail metric is tested using a non inferiority test and the significance level $\alpha$
        \item all superiority and non-inferiority tests are designed for the power level
         \begin{align*}
            \beta^*=\begin{cases} 
            \frac{\beta - \alpha}{(1 -\alpha)(G+1)} &   S>0\\ 
            \frac{\beta - \alpha}{(1 -\alpha)G}&   S=0 \end{cases}
        \end{align*} 
    \end{enumerate}
   then the error rates for Decision Rule \ref{dr2} are at most $\alpha$ and $\beta$.
\end{proposition}
\begin{proof} 
Let now $R_{(\mathcal{S})}$, $R_{(\mathcal{G})}$, $R_{(\mathcal{D})}$, and $R_{(\mathcal{Q})}$ denote the events that the composite success, guardrail, deterioration, and quality null hypotheses are rejected, respectively.
Furthermore, let $R_{(\mathcal{S})}^{(s)}$ and $R_{(\mathcal{D, S})}^{(s)}$ be the event that the superiority and deterioration null hypothesis for success metric $S$ are rejected, respectively. Similarly, let $R_{(\mathcal{G})}^{(g)}$ and $R_{(\mathcal{D, G})}^{(g)}$ be the event that the non-inferiority and deterioration null hypothesis for guardrail metric $g$ are rejected, respectively. We start by by proving that the type I error rate is bounded by $\alpha$ under the null hypothesis, where we start by using $\alpha^*=\frac{\alpha}{S}$. 
Denote the global null $\bxi = \btheta^{(\mathcal{S})} \in \Theta_0^{(\mathcal{S})}, \btheta^{(\mathcal{G})}\in \Theta^{(\mathcal{G})}_0$. We have under the global null
\begin{align*}
&\mathrm{Pr}\left(R_{(\mathcal{S})}\bigcap R_{(\mathcal{G})}\bigcap \lnot R_{(\mathcal{D})} \bigcap \lnot R_{(\mathcal{Q})} \Big|\bxi \right) =\\
&\mathrm{Pr}\left(\lnot R_{(\mathcal{D})} \bigcap \lnot R_{(\mathcal{Q})}   \mid R_{(\mathcal{S})}\bigcap R_{(\mathcal{G})}; \bxi \right) \times \\ &\mathrm{Pr}\left(R_{(\mathcal{S})}\bigcap R_{(\mathcal{G})}\Big| \bxi \right)
\end{align*}
By Lemmas \ref{lemma:sig_pow_intersec_union} and \ref{lemma:sig_pow_union_intersec} we have that $\mathrm{Pr}\left(R_{(\mathcal{S})} \mid \bxi\right)\leq \alpha$ and $\mathrm{Pr}\left(R_{(\mathcal{G})} \mid \bxi \right)\leq \alpha$, which means
\begin{align*}
&\mathrm{Pr}\left(R_{(\mathcal{S})}\bigcap R_{(\mathcal{G})} \mid \bxi \right)\leq \\
& \min \left( 
\mathrm{Pr}\left(R_{(\mathcal{S})} \mid \bxi\right),
\mathrm{Pr}\left(R_{(\mathcal{G})} \mid \bxi\right)
\right) \leq\alpha.
\end{align*}
This directly implies that 
\begin{align*}
   &\mathrm{Pr}\left(\lnot R_{(\mathcal{D})} \bigcap \lnot R_{(\mathcal{Q})}   \mid R_{(\mathcal{S})}\bigcap R_{(\mathcal{G})}; \bxi \right) \times \\& \mathrm{Pr}\left(R_{(\mathcal{S})}\bigcap R_{(\mathcal{G})}\Big| \bxi \right)\\
   &\leq \mathrm{Pr}\left(R_{(\mathcal{S})}\bigcap R_{(\mathcal{G})} \mid \bxi \right) \\
   &\leq \alpha.
\end{align*}
In other words, the effect of the deterioration tests can only make the false positive rate lower than $\alpha$.
This concludes the first part of the proof. 
\\~\\~\\
Second, for ease of reading we denote $\bkappa = \btheta^{(\mathcal{S})}\in \Theta_{\beta, \mathcal{A}}^{(\mathcal{S})}, \btheta^{(\mathcal{G})}\in \Theta_{\beta, \mathcal{A}}^{(\mathcal{G})}$ and $\bpsi = R_{(\mathcal{S})}\bigcap R_{(\mathcal{G})}; \bxi$. It remains to prove that 
\begin{align*}
    &\mathrm{Pr}\left(R_{(\mathcal{S})}\bigcap R_{(\mathcal{G})} \bigcap \lnot R_{(\mathcal{D})} \bigcap \lnot R_{(\mathcal{Q})}  \mid \bkappa \right)\geq 1-\beta.
\end{align*}
We have that
\begin{align*}
    &\mathrm{Pr}\left(R_{(\mathcal{S})}\bigcap R_{(\mathcal{G})} \bigcap \lnot R_{(\mathcal{D})} \bigcap \lnot R_{(\mathcal{Q})}  \mid \bkappa\right)=\\&
    \mathrm{Pr}\left( \lnot R_{(\mathcal{D})}\bigcap \lnot R_{(\mathcal{Q})}  \mid\bpsi\right) \times \mathrm{Pr}\left(R_{(\mathcal{S})}\bigcap R_{(\mathcal{G})} \mid \bkappa\right)
\end{align*}
By Proposition \ref{prop:alpha_beta_under_correction}, we know that 
$$\mathrm{Pr}\left(R_{(\mathcal{S})}\bigcap R_{(\mathcal{G})} \mid \bkappa\right) \geq 1- \beta$$ 
when each superiority and non-inferiority is design to achieve $\beta^*=\frac{\beta}{G+1}$.
Looking closer at the first factor we note that 
\begin{align*}
   &\mathrm{Pr}\left( \lnot R_{(\mathcal{D})}\bigcap \lnot R_{(\mathcal{Q})}  \mid \bpsi \right)  = \\
   &\mathrm{Pr}\Big( \bigcap_{s=1}^S \lnot R_{(\mathcal{D, S})}^{(s)}\cap \bigcap_{g=1}^G \lnot R_{(\mathcal{D, G})}^{(g)}\cap \\& \bigcap_{d=1}^D \lnot R_{(\mathcal{D})}^{(d)}\cap \bigcap_{q=1}^Q \lnot R_{(\mathcal{Q})}^{(q)} \mid \bpsi\Big)= \\
   & 1- \mathrm{Pr}\Big( \bigcup_{s=1}^S R_{(\mathcal{D, S})}^{(s)}\cup \bigcup_{g=1}^G  R_{(\mathcal{D, G})}^{(g)}\cup \\& \bigcup_{d=1}^D R_{(\mathcal{D})}^{(d)}\cup \bigcup_{q=1}^Q  R_{(\mathcal{Q})}^{(q)} \mid \bpsi\Big)\\
   & \geq 1-\sum_{s=1}^S\mathrm{Pr}\left( R_{(\mathcal{D, S})}^{(s)} \Big| \bpsi\right) - \sum_{g=1}^G \mathrm{Pr}\left( R_{(\mathcal{D, G})}^{(g)} \Big| \bpsi\right)
   - \\& \sum_{t=1}^D\mathrm{Pr}\left( R_{(\mathcal{D})}^{(t)} \Big| \bpsi\right)-
   \sum_{q=1}^Q\mathrm{Pr}\left( R_{(\mathcal{Q})}^{(q)} \Big| \bpsi\right)
\end{align*}
Conditioned on $\bpsi$ it is unlikely that all guardrail metrics and success metrics that are significantly non-inferior and superior, respectively, are simultaneously significantly inferior. However, we can bound all of these probabilities to $\leq \frac{1}{S+G+D+Q}\alpha_{-}$ as a worst case scenario. This implies that  
\begin{align*}
\sum_{s=1}^S\mathrm{Pr}\left( R_{(\mathcal{D, S})}^{(s)} \Big| \bpsi\right) &\leq \frac{S}{S+G+D+Q}\alpha_{-}\\
    \sum_{g=1}^G \mathrm{Pr}\left( R_{(\mathcal{D, G})}^{(g)} \Big| \bpsi\right)&\leq\frac{G}{S+G+D+Q}\alpha_{-}\\
   \sum_{t=1}^D\mathrm{Pr}\left( R_{(\mathcal{D})}^{(t)} \Big| \bpsi\right)&\leq\frac{D}{S+G+D+Q}\alpha_{-}\\
   \sum_{q=1}^Q\mathrm{Pr}\left( R_{(\mathcal{Q})}^{(q)} \Big| \bpsi\right)&\leq\frac{Q}{S+G+D+Q}\alpha_{-}. 
\end{align*}
Plugging this into the above yield
\begin{align*}
    &\mathrm{Pr}\left(R_{(\mathcal{S})}\bigcap R_{(\mathcal{G})} \bigcap \lnot R_{(\mathcal{D})}\bigcap \lnot R_{(\mathcal{Q})} \mid \bkappa\right)=\\
    &\mathrm{Pr}\left( \lnot R_{(\mathcal{D})}\bigcap \lnot R_{(\mathcal{Q})} \mid \bpsi \right)\mathrm{Pr}\left(R_{(\mathcal{S})}\bigcap R_{(\mathcal{G})} \mid \bkappa\right)\\
    &\geq (1-\alpha_{-})(1-\beta).
\end{align*}
To find the corrected $\beta$ for achieving the intended type II error rate,  we solve 
\begin{align*}
 (1 - \alpha_{-})(1 - \beta^*) &= 1-\beta 
 \Leftrightarrow
 \beta^* = \frac{\beta - \alpha_{-}}{1 - \alpha_{-}}.
\end{align*} 
Together this means that, if we design each superiority and non-inferiority test to obtain power $1-\beta^*$ where $\beta^*=\frac{\beta - \alpha_{-}}{(1 -\alpha_{-}) (G+1)}$, the type II error rate is bounded by $\beta$. This concludes the proof. 
\end{proof}
The condition that $\alpha_-<\beta$ should be interpreted as that is only possible to correct the false negative rate (and thus ensuring the intended power) for decision rules that include deterioration and/or quality tests as long as the intended false negative rate for the decision is larger than the intended false positive rate for those tests. This is quite natural, if we add a large enough chance of rejecting no deterioration or quality, even when there is no deteriration or problem with the quality, this will at some point (for some $\alpha_-$) limit our ability to find a positive decision, regardless of the sample size.  

For success and guardrail metrics, there is a dependency between the rejection of the deterioration test and the superiority or non-inferiority test for any given metric, respectively. That is, if the superiority or non-inferiority test is rejected, this affects the probability of the deterioration test to be rejected too. It is possible to utilize this dependency to improve the efficiency of Proposition \ref{prop:alpha_beta_under_correction_with_deterioration_and_quality} slightly by making additional assumptions about the relation between $\alpha$, $\alpha_-$, and $\beta$. See Appendix \ref{app:improving_the_prop} for details.

\section{Monte Carlo simulation study}\label{sec:MC}
In this section we run a simulation study to illustrate the empirical error rates for the multi-metric decision rules with and without the alpha and power corrections \footnote{Code for replication can be found at \url{https://github.com/MSchultzberg/Risk-management-paper-2024}.}. 
To make the simulation more relevant, all deterioration and quality tests use Group Sequential Tests (GST). All non-inferiority and superiority tests use fixed-horizon $z$-tests. See Appendix \ref{app:seq_test} for a discussion about combining sequential tests and fixed horizon tests for the same metric. For the GSTs, we analyze the results 10 times during the data collection at evenly spaced intervals. We generate data from a multivariate normal distribution, and treat the variance as known in the tests to be able to keep the sample size small. We use $S=G=5$ and $D=Q=2$, and repeat the simulation 100,000 times for each setting. In all scenarios, $\alpha^+=\alpha^-=0.05$ and $\beta=0.2$.
We compare three designs: no correction, only alpha correction, and Proposition \ref{prop:alpha_beta_under_correction_with_deterioration_and_quality} with Remark \ref{remark:power_dr2}. 
In the simulation study, we vary the following:
\begin{itemize}
    \item Hypothesis under which the simulation is performed
    \begin{itemize}
        \item $H_0$: the null of the non-inferiority and superiority tests
        \item Status quo: the null of the superiority and the alternative for the non-inferiority tests
        \item $H_1$: the alternative for non-inferiority and superiority tests
    \end{itemize}
    \item Dependence structure
    \begin{itemize}
        \item Independent: All metrics independent
        \item Dependent: All pairwise correlations 0.99
        \item Block 1: all guardrail metrics independent of each other and the success metrics, but all success metrics have pairwise correlation 0.99
        \item Block 2: all success metrics independent of each other and the guardrail metrics, but all guardrail metrics have pairwise correlation 0.99
    \end{itemize}
\end{itemize}
For all settings, all additional deterioration and quality metrics are generated as independent of each other and all other metrics with a zero effect.

The dependency structures are chosen to illustrate the most extreme situations
The naive "Only Alpha" correction is simply dividing alpha on the number of tests, i.e., implies $\alpha^*=\frac{\alpha}{S+G}=\frac{0.05}{10}$ and $\alpha^*_-=\frac{\alpha_-}{S+G+D+Q}=\frac{0.05}{14}$. 

\subsection{Results}
For convenience, the results are split by scenario: $H_0$, Status quo, and $H_1$. We present the rejection rates for the following groups of tests:
\begin{itemize}
    \item $R_\mathcal{S}$: superiority tests for success metrics, a rejection means at least one is reject in that replication.
    \item $R_\mathcal{G}$: non-inferiority tests for guardrail metrics, a rejection means all are rejected in that replication.
    \item $R_{\mathcal{D},\mathcal{S}}\bigcup R_{\mathcal{D}, \mathcal{G}}$: inferiority tests for success and guardrail metrics, a rejection means at least one is rejected in that replication.
    \item $R_\mathcal{D}\bigcup R_{\mathcal{Q}}$: tests for deterioration and quality metrics, a rejection means at least one is rejected in that replication.
\end{itemize}
We also present the rejection rates for the three decision rules we develop in the paper, however, note that the design is always under Proposition \ref{prop:alpha_beta_under_correction_with_deterioration_and_quality}. 

\subsubsection{Results under the global $H_0$}
Table \ref{tab:sim_h0} displays the result under the the null hypotheses of the non-inferiority and the superiority tests. 
As expected, all three decision rules are conservative under all settings. This is explained by two things. First, the power of the deterioration test on the guardrail metrics is very high under the null hypothesis of the non-inferiority tests. Second, the probability of rejecting all guardrail metrics simultaneously is very low, with the exception of when all guardrail metrics are strongly correlated, which is confirmed for the Dependent and Block 1 covariance settings. 
\begin{table*}[t]
\begin{tabular}{rrlllll}
  \hline
Covariance & Correction & $R_\mathcal{S}$ & $R_\mathcal{G}$ & $R_\mathcal{D,S}\bigcup R_\mathcal{D,G}$ & $R_\mathcal{D}\bigcup R_\mathcal{Q}$ & Decision Rule 2 \\ 
  \hline
\multirow{3}{*}{Independent} & None & 0.227 & 0.000 & 0.999 & 0.186 & 0.000 \\ 
   & Only Alpha & 0.049 & 0.000 & 0.866 & 0.014 & 0.000 \\ 
   & Prop. 4.1 & 0.049 & 0.000 & 0.999 & 0.014 & 0.000 \\ 
   \cline{2-7}
  \multirow{3}{*}{Dependent} & None & 0.063 & 0.040 & 0.773 & 0.186 & 0.032 \\ 
   & Only Alpha & 0.014 & 0.040 & 0.376 & 0.014 & 0.014 \\ 
   & Prop. 4.1 & 0.014 & 0.040 & 0.771 & 0.014 & 0.014 \\ 
   \cline{2-7}
  \multirow{3}{*}{Block 1} & None & 0.227 & 0.040 & 0.825 & 0.186 & 0.006 \\ 
   & Only Alpha & 0.050 & 0.040 & 0.388 & 0.014 & 0.002 \\ 
   & Prop. 4.1 & 0.050 & 0.040 & 0.776 & 0.014 & 0.002 \\ 
  \cline{2-7}
  \multirow{3}{*}{Block 2} & None & 0.063 & 0.000 & 0.999 & 0.186 & 0.000 \\ 
   & Only Alpha & 0.014 & 0.000 & 0.865 & 0.014 & 0.000 \\ 
   & Prop. 4.1 & 0.014 & 0.000 & 0.999 & 0.014 & 0.000 \\ 
   \hline
\end{tabular}
    \caption{Simulation results of the type I error rates under the null hypotheses of the non-inferiority and the superiority tests, respectively. The rejection of the deterioration, non-inferiority, and superiority tests are presented along with the rejection rate of Decision Rule \ref{dr2}.}
    \label{tab:sim_h0}
\end{table*}
As stated before, the $H_0$ scenario is arguably of little practical relevance, because the null hypothesis of the non-inferiority is decided by the experimenter. A more relevant scenario, in our opinion, is the status quo scenario for which the next section displays the results.

\subsubsection{Results under status quo}
Table \ref{tab:sim_statusQuo} displays the result under the the alternative hypothesis of the non-inferiority and the null hypothesis of the superiority tests. The results clearly show that Proposition \ref{prop:alpha_beta_under_correction_with_deterioration_and_quality} is the correction that has the highest type I error rate, while not crossing the intended $\alpha=0.05$. The Only Alpha correction bounds the type I error rate, but is conservative due to the lack of simultaneous power for the guardrail metrics. As expected, the addition of the deterioration and quality tests does make the rejection rate more conservative, but on a magnitude that has little practical relevance, especially as compared to, e.g., the impact of including beta corrections in the analysis.     
\begin{table*}[t]
    \centering
\begin{tabular}{rrlllll}
  \hline
Covariance & Correction & $R_\mathcal{S}$ & $R_\mathcal{G}$ & $R_\mathcal{D,S}\bigcup R_\mathcal{D,G}$ & $R_\mathcal{D}\bigcup R_\mathcal{Q}$ & Decision Rule 2 \\ 
  \hline
\multirow{3}{*}{Independent} & None & 0.227 & 0.328 & 0.399 & 0.184 & 0.047 \\ 
   & Only Alpha & 0.050 & 0.328 & 0.035 & 0.015 & 0.016 \\ 
   & Prop. 4.1 & 0.050 & 0.877 & 0.035 & 0.015 & 0.042 \\ 
   \cline{2-7}
  \multirow{3}{*}{Dependent} & None & 0.062 & 0.765 & 0.068 & 0.184 & 0.051 \\ 
   & Only Alpha & 0.013 & 0.765 & 0.006 & 0.015 & 0.013 \\ 
   & Prop. 4.1 & 0.013 & 0.966 & 0.006 & 0.015 & 0.013 \\ 
   \cline{2-7}
  \multirow{3}{*}{Block 1} & None & 0.228 & 0.765 & 0.272 & 0.184 & 0.113 \\ 
   & Only Alpha & 0.048 & 0.765 & 0.023 & 0.015 & 0.036 \\ 
   & Prop. 4.1 & 0.048 & 0.966 & 0.023 & 0.015 & 0.045 \\ 
   \cline{2-7}
  \multirow{3}{*}{Block 2} & None & 0.062 & 0.329 & 0.272 & 0.184 & 0.016 \\ 
   & Only Alpha & 0.013 & 0.329 & 0.023 & 0.015 & 0.004 \\ 
   & Prop. 4.1 & 0.013 & 0.878 & 0.023 & 0.015 & 0.011 \\ 
   \hline
\end{tabular}
    \caption{Simulation results of the rejection rates under the alternative hypothesis of the non-inferiority and the null hypothesis of superiority tests. The rejection rates of the deterioration, non-inferiority, and superiority tests are presented along with the rejection rate of Decision Rule \ref{dr2}.}
    \label{tab:sim_statusQuo}
\end{table*}

\subsubsection{Results under the global $H_1$}
Table \ref{tab:sim_h1} displays the result under the the alternative hypotheses of the non-inferiority and the superiority tests.
In this setting, the need for the power correction imposed by Proposition \ref{prop:alpha_beta_under_correction_with_deterioration_and_quality} is clear. For the other two corrections, the rates with which the guardrail metrics are simultaneously significantly non-inferior ($R_\mathcal{G}$) do not reach the intended power of 80\%, which implies that neither can the decision rules. In the settings where the guardrail metrics are all independent (Independent, and Block 1), the rejection rate under no correction and the Only Alpha correction are as low as 30\% for both decision rules---less than half of the intended power.  
\begin{table*}[t]
\centering
\begin{tabular}{rrlllll}
  \hline
Covariance & Correction & $R_\mathcal{S}$ & $R_\mathcal{G}$ & $R_\mathcal{D,S}\bigcup R_\mathcal{D,G}$ & $R_\mathcal{D}\bigcup R_\mathcal{Q}$ & Decision Rule 2\\ 
  \hline
\multirow{3}{*}{Independent} & None & 1.000 & 0.328 & 0.229 & 0.184 & 0.255 \\ 
   & Only Alpha & 1.000 & 0.328 & 0.018 & 0.014 & 0.323 \\ 
   & Prop. 4.1 & 1.000 & 0.874 & 0.017 & 0.014 & 0.859 \\ 
   \cline{2-7}
  \multirow{3}{*}{Dependent} & None & 0.832 & 0.766 & 0.064 & 0.184 & 0.618 \\ 
   & Only Alpha & 0.832 & 0.766 & 0.005 & 0.014 & 0.756 \\ 
   & Prop. 4.1 & 0.980 & 0.966 & 0.005 & 0.014 & 0.952 \\ 
   \cline{2-7}
  \multirow{3}{*}{Block 1} & None & 1.000 & 0.766 & 0.067 & 0.184 & 0.616 \\ 
   & Only Alpha & 1.000 & 0.766 & 0.005 & 0.014 & 0.756 \\ 
   & Prop. 4.1 & 1.000 & 0.966 & 0.005 & 0.014 & 0.952 \\ 
   \cline{2-7}
  \multirow{3}{*}{Block 2} & None & 0.832 & 0.329 & 0.227 & 0.184 & 0.213 \\ 
   & Only Alpha & 0.832 & 0.329 & 0.017 & 0.014 & 0.270 \\ 
   & Prop. 4.1 & 0.980 & 0.874 & 0.017 & 0.014 & 0.842 \\ 
   \hline
\end{tabular}
\caption{Simulation results of the type I error rate under the alternative hypotheses of the non-inferiority and the superiority tests, respectively. The rejection rates of the deterioration, non-inferiority, and superiority tests are presented along with the rejection rate of Decision Rule \ref{dr2}.}\label{tab:sim_h1}
\end{table*}

The results under $H_1$  show that Proposition \ref{prop:alpha_beta_under_correction_with_deterioration_and_quality} bounds the error rates also in the worst-case scenarios, which leads to a higher power than desired in the best-case scenarios. For example, Proposition \ref{prop:alpha_beta_under_correction_with_deterioration_and_quality}  ensures that even under the Block 2 covariance matrix, where all guardrail metrics are independent and all success metrics are dependent, the power is above 80\%. However, to ensure that level of power under the worst-case scenario, the power instead exceeds 94\% in the best-case scenario. 

\section{Discussion and conclusions}
In this paper, we introduce a decision rule framework that uses the results of multiple frequentist hypothesis tests in an experiment to ultimately make one single decision. The primary area of application of our theory is online controlled experiments, known as A/B tests, that technology companies use to run randomized controlled trials at scale to inform their product decision making. Decision rules help move the focus from the results of individual hypothesis tests to what really matters in the end: the overall conclusion of whether the tested change is good enough to release more widely. By carefully tailoring the decision rules, modern-day experimentation platforms can leverage decision rules to standardize decision making. For example, experimentation platforms can enforce that a new change should only be released widely if there is evidence that the change does not negatively impact important business metrics. Our framework clearly lays out the appropriate adjustments to the experiment design to ensure that the tolerated type I and type II error rates are not exceeded at the experiment level. Experimentation is primarily about understanding and managing risks, and with decision rules we put the risks that are relevant to the business front and center.

Our decision rules use measured outcomes, known as metrics, that we classify into success metrics, guardrail metrics, deterioration metrics and quality metrics. Success metrics are metrics where we want to see an improvement, and guardrail metrics are metrics where we want to make sure no regression happens. We show that for a decision rule that includes both success metrics and guardrail metrics, we only need to adjust the type I error $\alpha$ used in tests for the number of success metrics. This result stands in contrast to the wide use of multiple correction methods, which generally do not take into account the role of the metrics they adjust for. For guardrail metrics, we introduce the concept of $\beta$, or type II, corrections. We show that decision rules may be grossly under-powered without correcting the type II error the experiment is designed to achieve for the number of guardrail metrics. Finally, we introduce deterioration tests and quality tests, such as sample ratio mismatch and tests for pre-exposure bias, to the decision rule, and show how these tests affect the type I and type II errors. We combine all of our findings into the decision rule that is used by Spotify's experimentation platform, which includes tests for superiority, non-inferiority, inferiority, and quality. In a Monte Carlo simulation study, we illustrate how the design that our decision rule implies bounds the risks appropriately under various data-generating processes. 

There are many ways to construct decision rules incorporating many metrics, and many ways to control the type I and II errors for them. The decision rules and their corresponding $\alpha$ and $\beta$ corrections that we present in this paper are motivated mainly by simplicity. The straightforward Bonferroni-like corrections are easy to implement, and importantly, easy to explain and understand for most experimenters. In our experience, simplicity is often more important than finding the statistically optimal solution when it comes to experimentation methods that are to be used by many experimenters in large companies. 

Two important aspects of more sophisticated multiple-correction methods e.g., Holm and Hommel, is that they tend to complicate sample size calculations and they do not support straightforward calculation of confidence intervals, if at all. The net effect of using these more sophisticated correction methods is, in our experience, negative as the increased complexity leads to poorer planning of experiments and results that are harder to interpret. An appealing compromise for A/B tests is Nyholt's method that adjusts for the efficient number of independent tests \citep{nyholt2004}. In Section \ref{sec:app:nyholt} in the Appendix, we discuss how Nyholt's method can be used and similarly to \cite{Kong2004} we find that for a small number of metrics (5 success and 5 guardrail) the improvement in efficiency is marginal unless the correlation is strong. However, the improvement becomes substantial when the number of correlated metrics is large. In addition, methods other than sophisticated multiple-testing corrections can help increase efficiency, for example multivariate variance reduction via regression adjustment.

For success metrics it is possible to replace the inferiority tests used in addition to the superiority tests, with non-inferiority tests. Strictly speaking, failing to reject the null of the inferiority test provides no evidence for the null hypothesis, and thus a more formal approach would be to require non-inferiority of all success metrics too. The reason for why inferiority tests are proposed here instead of non-inferiority tests is pragmatic. Using non-inferiority tests would require specifying a non-inferiority margin in addition to the hypothetical effect which would complicate the experiment setup for non-technical experimenters. Any deterioration metric could be replaced by a guardrail metric if the higher standard of non-inferiority tests is desired. Again, we find that having a few business important metrics as automatic deterioration metrics is a reasonable compromise: the experimenter does not need to configure anything for these metrics (no non-inferiority margin), but the company will detect major regressions to these metrics. This is a great example of the trade-off between rigour and not adding too much friction to the product development.   

In our experience, it is always the case that there is information outside the test that decision makers include in their product decision. However, to benefit from A/B testing as a risk management apparatus, we think it is important to move as much as possible into the experiment and explicitly define the decision rule with all the possible aspects in mind. As this paper shows, including caveats for shipping in terms of quality and deterioration tests affect the error rates of the decision. In practice, experimenters hesitate to include certain metrics as guardrail metrics because the sample size needed to power the experiment is too high. If these excluded metrics are still used in the decision making, the effect of the exclusion is only poor management of risk. Including as much as possible in the experiment and its decision rule means a better understanding of the risks you are facing for the the decision you ultimately need to make.

\pagebreak
\onecolumn
\begin{appendix}

\section{Improving the efficiency of Proposition \ref{prop:alpha_beta_under_correction_with_deterioration_and_quality} with additional assumptions}\label{app:improving_the_prop}
It is possible to improve the efficiency by making assumptions that makes the inferiority test rejection region not overlap with the rejection region of the superiority or non-inferiority test, respectively, for success metrics and guardrail metrics. 
In the sections below, we describe how deterioration tests affect the type I and type II error rates for superiority and non-inferiority tests and under what conditions the rates are not affected. We then use this information to update Proposition \ref{prop:alpha_beta_under_correction_with_deterioration_and_quality}. 

\subsection{Type I and type II error rates for a success metric with superiority and deterioration tests}
To be specific about the consequences of testing a success metric with an additional deterioration test, we require some new notation. We focus on a single success metric and how this metric contributes to the shipping decision. There are now two tests and four hypotheses for the success metric, defined as:
\begin{align*}
    H_0^{(\mathcal{S})} : \delta\leq 0 &\text{ versus } H_\mathcal{A}^{(\mathcal{S})}: \delta \geq 0 \tag{Superiority test}\\
    H_0^{(-\mathcal{S})} : \delta\geq 0 &\text{ versus } H_\mathcal{A}^{(-\mathcal{S})}: \delta \leq 0 \tag{Inferiority test}
\end{align*}
The deterioration test is here an inferiority test that mirrors the main superiority test, but with the hypotheses reversed.
Because of these two tests, there are now two possible false positive results for this metric: a false positive inferiority test, and a false positive superiority test. The false positive risk in terms of the shipping decision for a set of success metrics with superiority and inferiority tests is given by the probability that the superiority test rejects the null hypothesis while, simultaneously, the inferiority test does not. The additional criterion of non-significant inferiority tests can only make the superiority test more conservative, i.e., reduce the false positive rate and the true positive rate of the shipping decision. However, it is easy to restrict the rejection regions of the deterioration test and the superiority test such that the deterioration test cannot affect the rejection rate of the superiority test under any hypothesis. 
Lemma \ref{lemma:fpr_and_power_success_with_deterioration} formalizes this. 
\begin{lemma}\label{lemma:fpr_and_power_success_with_deterioration}
    Suppose $\alpha$ and $\alpha_{-}$ are the nominal significance levels used for the superiority and inferiority tests, respectively, for a single success metric. Let both tests be based on the test statistic $T:\mathbb{R}^n\mapsto\mathbb{R}$. Let $R$ represent the event that the inferiority test is accepted and the superiority test is rejected. If $1-\alpha > \alpha_{-}$, then $Pr(R \mid \btheta \in \Theta_0) = \alpha$ and $Pr(R \mid \btheta \in \Theta_\mathcal{A}) = 1-\beta$, i.e., the type I and type II error rates for the shipping decision is unaffected by the inferiority test. 
\end{lemma}
\begin{proof}
Let $R_{\mathcal{S}}=\{x\in\mathbb{R}:T(x)>c_{1-\alpha}\}$ be the rejection region for the superiority test, where $c_t=F_T^{-1}(t)$ is the critical value and $F_T^{-1}$ is the inverse cumulative distribution function of $T$. Let $R_{\mathcal{D}}=\{x\in\mathbb{R}:T(x)<c_{\alpha_-}\}$ be the rejection region for the inferiority test. We have $1-\alpha = F(c_{1-\alpha})$ and $\alpha_{-} = F(c_{\alpha_-})$.

If $1-\alpha > \alpha_{-}$ then $c_{1-\alpha} > c_{\alpha_-}$ because $F$ is a monotonically increasing function. Suppose $x \in R_\mathcal{S}$ then $T(x) > c_{1-\alpha} > c_{\alpha_-}$ which means that $x \in \lnot R_\mathcal{D}$, so $R_S \subseteq \lnot R_\mathcal{D}$. Therefore $R = R_\mathcal{S} \cap \lnot R_\mathcal{D} = R_\mathcal{S}$, and we get the desired result

\begin{align*}
Pr(x \in R \mid \theta \in \Theta_{0}^{(S)})& = Pr(x \in R_\mathcal{S}\cap \lnot R_\mathcal{D} \mid \theta \in \Theta_{0}^{(S)})\\
&=Pr(x \in R_\mathcal{S} \mid \theta \in \Theta_{0}^{(S)})\\
&=\alpha.
\end{align*}
Similarly,
\begin{align*}
Pr(x \in R \mid \theta \in \Theta_{\mathcal{A}}^{(S)})& = Pr(x \in R_\mathcal{S}\cap \lnot R_\mathcal{D} \mid \Theta_{\mathcal{A}}^{(S)})\\
&=Pr(x \in R_\mathcal{S} \mid \theta \in \Theta_{\mathcal{A}}^{(S)})\\
&=\beta.
\end{align*}
This concludes the proof.
\end{proof}
Lemma \ref{lemma:fpr_and_power_success_with_deterioration} says that if the same test statistic is used for the superiority test and the deterioration test, and $\alpha_-<1-\alpha$, the rejection regions cannot overlap. This is somewhat trivial for the success metric case. In the next section, we present similar results for guardrail metrics with non-inferiority and deterioration tests.  

\subsection{Type I and type II error rates for a guardrail metric with non-inferiority and deterioration tests}
For guardrail metrics, we can reason analogously to the success metrics case, but there are additional nuances due to the potential overlap between rejection regions of the deterioration and non-inferiority tests. There are four relevant hypotheses given by  
\begin{align*}
    H_0^{(\mathcal{G})} : \delta\leq -NIM &\text{ versus } H_\mathcal{A}^{(\mathcal{G})}: \delta \geq -NIM \tag{Non-inferiority test}\\
    H_0^{(-\mathcal{G})} : \delta\geq 0 &\text{ versus } H_\mathcal{A}^{(-\mathcal{G})}: \delta \leq 0 \tag{Inferiority test}
\end{align*}
Note that, contrary to the success metric case, $H_\mathcal{A}^{(\mathcal{G})} \nrightarrow \lnot H_\mathcal{A}^{(-\mathcal{G})}$ even if $\alpha_-<1-\alpha$. That is, the treatment can be non-inferior to control while it is also inferior to control. This happens if the NIM is chosen such that it has power close to one for a given design. If the power for the non-inferiority test is close to one, the power for the deterioration test can, at the same time, also be high in scenarios where the true treatment effect is between $-NIM$ and $0$.
The experimenter is free to set the NIM, and making this choice is a common challenge for experimenters. For this reason, we believe that testing guardrail metrics with both non-inferiority and inferiority tests in experimentation platforms is a way to guard against NIMs that are unnecessarily large.

\begin{lemma}\label{lemma:fpr_and_power_non_inf_deterioration}
    Suppose $\alpha_{+}$ and $\alpha_{-}$ are the nominal significance levels used for the non-inferiority and inferiority test, respectively, for a single guardrail metric. Let the non-inferiority test be based on the test statistic $T_\delta:\mathbb{R} \mapsto\mathbb{R}$, and the inferiority test be based on the test statistic $T:\mathbb{R} \mapsto \mathbb{R}$, with the property $T_\delta = T + m(\delta)$ where $m: \mathbb{R} \mapsto \mathbb{R}$, $m(0) = 0$, $m(-\delta) = -m(\delta)$ and represents the shift in $T(x)$ under the alternative. Let the non-inferiority test have type II error $Pr(R \mid \theta \in \{0\}) = 1-\beta$ when the non-inferiority margin is $\nu$. Let $R$ represent the event that the inferiority test is accepted and the non-inferiority test is rejected i.e. the shipping decision. If $\alpha_{-} < \beta$, then $Pr(R \mid \theta \in \Theta_0^{\mathcal{G}}) = \alpha_{+}$ and $Pr(R \mid \theta \in \Theta_{\mathcal{A}}^{\mathcal{G}}) = 1-\beta$, i.e., the type I and type II error rates for the shipping decision is unaffected by the inferiority test.
\end{lemma}
\begin{proof}
Let $F(t)$ be the distribution function for $T(x)$. First, we show how to express the critical value of the guardrail in terms of $F(t)$. We want to solve $Pr(T_{-\nu}(x) > c^{(\mathcal{G})}_{1-\alpha_+}) = 1 - \alpha_{+}$, using that $T_\delta(x) = T(x) + m(\delta)$  we get that $Pr(T(x) > m(\nu) + c^{(\mathcal{G})}_{1-\alpha_+}) = 1 - \alpha_{+}$ so the critical value of the guardrail is simply $c^{(\mathcal{G})}_{1-\alpha_+} = -m(\nu) + c_{1-\alpha_+}$ where $c_{1-\alpha_+} = F^{-1}(1-\alpha_{+})$.

Let $R_{\mathcal{G}} = \{x \in \mathbb{R} : T_{-\nu}(x) > -m(\nu) + c_{1-\alpha_+}\}$ be the rejection region for the non-inferiority test (using the previous result for the critical value), similarly, let $R_{\mathcal{D}} = \{x \in \mathbb{R}:T(x) < c_{\alpha_-}\}$ be the rejection region for the inferiority test.

We now show the main result. Suppose $x \in R_{\mathcal{G}}$, we then have

\begin{align*}
    x \in  R_{\mathcal{G}} & \Rightarrow T_{-\nu}(x) > -m(\nu) + c_{1-\alpha_+} \\
    & \Leftrightarrow T(x) - m(\nu) > c_{1-\alpha_+} - m(\nu)\\
    & \Leftrightarrow T(x) > c_{1-\alpha_+}.
\end{align*}
We now need to show that $\alpha_- < \beta \Rightarrow T(x) > c_{1-\alpha_+} > c_{\alpha_-}$ so that $x$ is also in the acceptance region $\lnot R_{\mathcal{D}}$ of the inferiority test,

\begin{align*}
    \alpha_- & < \beta \\
    Pr(T(x) < c_{\alpha_-} \mid \theta \in \{0\}) & < Pr(T_{-\nu}(x) \\&< - m(\nu) + c_{1-\alpha_+} \mid \theta \in \{0\})\\
    Pr(T(x) < c_{\alpha_-} \mid \theta \in \{0\}) & < Pr(T(x) < c_{1-\alpha_+} \mid \theta \in \{0\})\\
    F(c_{\alpha_-} \mid \theta \in \{0\}) & < F(c_{1 - \alpha_+} \mid \theta \in \{0\})\\
    c_{\alpha_-} & < c_{1 - \alpha_+}.\\
\end{align*}
This means that $T(x) > c_{\alpha_-}$ so that $x \in \lnot R_{\mathcal{D}}$ and $R_\mathcal{G} \subseteq \lnot R_\mathcal{D}$. Therefore $R = R_\mathcal{G} \cap \lnot R_\mathcal{D} = R_\mathcal{G}$. Similar to the success metric case, this means that 

\begin{align*}
Pr(x \in R \mid \theta \in \Theta_{0}^{(G)})& = Pr(x \in R_\mathcal{G}\cap \lnot R_\mathcal{D} \mid \theta \in \Theta_{0}^{(G)})\\
&=Pr(x \in R_\mathcal{G} \mid \theta \in \Theta_{0}^{(G)})\\
&=\alpha_{+}.
\end{align*}
Similarly,
\begin{align*}
Pr(x \in R \mid \theta \in \Theta_{\mathcal{A}}^{(G)})& = Pr(x \in R_\mathcal{G}\cap \lnot R_\mathcal{D} \mid \Theta_{\mathcal{A}}^{(G)})\\
&=Pr(x \in R_\mathcal{G} \mid \theta \in \Theta_{\mathcal{A}}^{(G)})\\
&=\beta.
\end{align*}
\end{proof}

\begin{figure}[h]\centering
\begin{tikzpicture}
            \begin{axis}[
              no markers, domain=-7:4, samples=100,
              axis lines*=left,
              xlabel=$z$,
              every axis y label/.style={at=(current axis.above origin),anchor=south},
              every axis x label/.style={at=(current axis.right of origin),anchor=west},
              height=5cm, width=10cm,
              xtick={-3.9712, 0}, xticklabels={{-NIM},0}, ytick=\empty,
              enlargelimits=false, clip=false, axis on top,
              grid = major
              ]
              \addplot [fill=cyan!20, draw=none, domain=-2.326:4] {gauss(0,1)} 
              \closedcycle;
               \addplot [fill=red!80, draw=none, domain=-5:-2.326] {gauss(0,1)} 
              \closedcycle;
               \addplot [fill=green!80, draw=none, domain=-2.326:4] {gauss(-3.9712,1)} 
              \closedcycle;
              \addplot [very thick,cyan!50!black] {gauss(-3.9712,1)};
              \addplot [very thick,cyan!50!black] {gauss(0,1)};
              \draw [yshift=-0.6cm, latex-latex](axis cs:-3.9712,0) -- node [fill=white] {$3.9712\sigma$} (axis cs:0,0);
            \node[label={$H_0^+,\, H_1^-$}] at (290,380) {};
            \node[label={$H_1^+,\, H_0^-$}] at (690,380) {};
            \end{axis}
        \end{tikzpicture}
    \caption{Illustration of the relation between rejection regions of the non-inferiority test and the inferiority test when $\beta=\alpha_{-}=0.01$. The green area is the rejection region for the non-inferiority test, the red area is the rejection region for the inferiority test, and the blue area is the power ($1-\beta$) of the non-inferiority test under the alternative $\delta=0$.}
    \label{fig:non-inf-with-inferiority}
\end{figure}
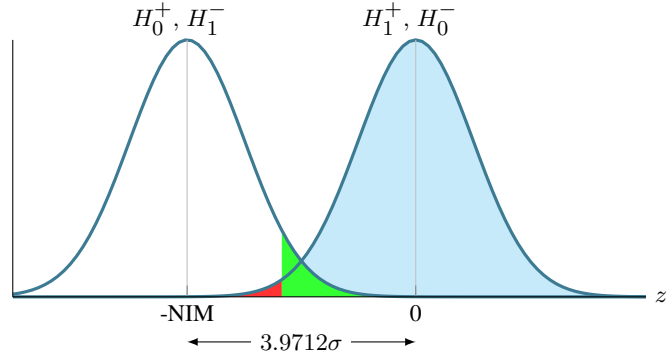
The fact that $\beta \geq \alpha_{-}$ is a sufficient additional criterion for the rejection regions to not overlap in the non-inferiority case is perhaps most easily seen graphically. Figure \ref{fig:non-inf-with-inferiority} displays the relation between the regions when $\beta=\alpha_{-}$. From the figure it is clear that the regions would overlap for a fixed $\beta$ if $\alpha_{-}$ increases. Analogously, a decrease in $\beta$ would make the regions overlap for a fixed $\alpha_{-}$. Note that as long as $\alpha_{-}\leq 1- \alpha_{+}$, the value of $\alpha_{+}$ does not affect the overlap. This follows since if $\alpha_{+}$ changes, the implied NIM for which the non-inferiority test has power $1-\beta$ changes accordingly.

If the NIM is chosen such that the non-inferiority test has high power, the deterioration test almost always rejects under $H^{(\mathcal{G})}_0$. In this case, the type I error rate of Decision rule \ref{dr1} is close to zero. However, in our opinion, this is of little practical relevance as the NIM, which determines $H^{(\mathcal{G})}_0$, is unlikely to be true in practice. It is selected because it is the threshold for how much deterioration we are willing to accept, and is not a scenario the experimenter expects to see or even necessarily finds likely.

\subsection{An updated version of Proposition \ref{prop:alpha_beta_under_correction_with_deterioration_and_quality}}
Here we improve the $\alpha$ and $\beta$ corrections slightly by enforcing non-overlapping rejections regions. 
\begin{proposition}
    Let $\alpha$ and $\beta$ be the nominal type I and type II rates that the decision rule should satisfy, and let there be $G$ guardrail metrics, $S$ success metrics, $D$ deterioration tests in addition to those for the success and guardrail metrics, and $Q$ quality metrics/tests. Let $\phi=\frac{D+S+Q-1}{S+G+D+Q}\alpha_{-}$. Moreover, for the guardrail metrics, let the non-inferiority and deterioration tests be based on the test statistics $T_\delta:\mathbb{R}^n\mapsto\mathbb{R}$, and $T:\mathbb{R}^n\mapsto\mathbb{R}$, where $T_\delta = T + m(\delta)$. If   
    \begin{enumerate}
        \item $\alpha_-$ is chosen such that $\alpha_- \leq 1-\alpha$, and $\alpha_- \leq \beta$
        \item each deterioration test is performed with significance level 
        $\alpha_-^*= \frac{\alpha_-}{S+G+D+Q}$
        \item each success metric is tested using a superiority test using the significance level 
        \begin{align*}
            \alpha^*=\begin{cases} 
            \frac{\alpha}{(1- \frac{1}{S+G+D+Q}\alpha_-)S} & \text{ if } S+D\geq 2\\ 
            \frac{\alpha}{S}& \text{ otherwise } \end{cases}
        \end{align*} 
        \item each guardrail metric is tested using a non-inferiority test and the significance level $\alpha$
        \item all superiority and non-inferiority tests are designed for the power level
        \begin{align*}
            \beta^*=\begin{cases} 
            \frac{\beta - \phi}{(1 -\phi)(G+1)} & \text{ if } S>0\\ 
            \frac{\beta - \phi}{(1 -\phi)G}& \text{ if } S=0 \end{cases}
        \end{align*} 
    \end{enumerate}
   then the error rates for Decision Rule \ref{dr2} are at most $\alpha$ and $\beta$.
\end{proposition}
\begin{proof}
 The proof is almost identical to the proof of Proposition \ref{prop:alpha_beta_under_correction_with_deterioration_and_quality}. Now we have under the global null
\begin{align*}
&\mathrm{Pr}\left(R_{(\mathcal{S})}\bigcap R_{(\mathcal{G})}\bigcap \lnot R_{(\mathcal{D})} \bigcap \lnot R_{(\mathcal{Q})} \Big|\bxi \right)\\
&\mathrm{Pr}\left(\lnot R_{(\mathcal{D})} \bigcap \lnot R_{(\mathcal{Q})}   \mid R_{(\mathcal{S})}\bigcap R_{(\mathcal{G})}; \bxi \right)\mathrm{Pr}\left(R_{(\mathcal{S})}\bigcap R_{(\mathcal{G})}\Big| \bxi \right)
\end{align*}
Which implies
\begin{align*}
    &\mathrm{Pr}\left(\lnot R_{(\mathcal{D})}\bigcap \lnot R_{(\mathcal{Q})}  \Big|\underbrace{R_{(\mathcal{S})}\bigcap R_{(\mathcal{G})}; \bxi}_{\bpsi} \right)= \\
    &\mathrm{Pr}\left(\lnot \left[
      R_{(\mathcal{D, S})}^{(1)}\cup \dots \cup  R_{(\mathcal{D, S})}^{(S)}\cup \lnot R_{(\mathcal{D, G})}^{(1)}\cup \dots \cup   R_{(\mathcal{D, G})}^{(G)}  \cup  R_{(\mathcal{D})}^{(1)}\cup \dots \cup   R_{(\mathcal{D})}^{(D)} 
    \right]\bigcap \lnot \left[ R_{(\mathcal{Q})}^{(1)} \cup\dots \cup R_{(\mathcal{Q})}^{(Q)} \right] \Big|\bpsi \right)= \\
   & \mathrm{Pr}\left(  \lnot R_{(\mathcal{D, S})}^{(1)}\cap \dots \cap \lnot R_{(\mathcal{D, S})}^{(S)}\cap \lnot R_{(\mathcal{D, G})}^{(1)}\cap \dots \cap  \lnot R_{(\mathcal{D, G})}^{(G)}  \cap \lnot R_{(\mathcal{D})}^{(1)}\cap \dots \cap  \lnot R_{(\mathcal{D})}^{(D)}\cap \lnot R_{(\mathcal{Q})}^{(1)} \cap\dots \cap \lnot R_{(\mathcal{Q})}^{(Q)} \Big| \bpsi\right)\leq \\
    &\min\left( \min_{s} \mathrm{Pr}\left(\lnot R_{(\mathcal{D,S})}^{(s)} \mid \bpsi \right) ,\min_g \mathrm{Pr}\left(\lnot R_{(\mathcal{D,G})}^{(g)} \mid \bpsi \right),\min_t \mathrm{Pr}\left(\lnot R_{(\mathcal{D})}^{(t} \mid \bpsi \right), \min_q \mathrm{Pr}\left(\lnot R_{(\mathcal{Q})}^{(q} \mid \bpsi \right) \right)=\\
    &\min(1-\alpha^*_-, 1, 1-\alpha^*_-,1-\alpha^*_-) = 1-\alpha^*_- = 1-\frac{1}{S+G+D+Q}\alpha_-,  
\end{align*}
where, despite Lemma \ref{lemma:fpr_and_power_success_with_deterioration}, $\mathrm{Pr}\left(\lnot R_{(\mathcal{D,S})}^{(s)} \mid \bpsi \right) = 1- \alpha_-^*$ and not $1$. This is because conditioning on $R_{(\mathcal{S})}$ only ensures that at least one success metric is significantly superior. Since we aim to bound the risk, it is enough if we can prove the proposition under the assumption that all $S$ metrics have $1-\alpha_-^*$ chance of deteriorating.
Together this means that 
\begin{align*}
   \mathrm{Pr}\left(R_{(\mathcal{S})}\bigcap R_{(\mathcal{G})}\bigcap \lnot R_{(\mathcal{D})}\bigcap \lnot R_{(\mathcal{Q})} \Big|\bxi \right) &\leq (1- \frac{1}{S+G+D+Q}\alpha_{-})\alpha < \alpha.
\end{align*}
We can adjust our correction of $\alpha$ to make the test less conservative by solving 
\begin{equation*}
    (1- \frac{1}{S+G+D}\alpha_{-})\alpha^* = \alpha
\end{equation*}
which implies that we should correct $\alpha$ according to
\begin{align*}
    \alpha^* = \frac{\alpha}{(1- \frac{1}{S+G+D}\alpha_{-})S},
\end{align*}
to minimize the conservativeness without any assumptions on the data generating process.

for the type II error rate we follow the proof of Proposition \ref{prop:alpha_beta_under_correction_with_deterioration_and_quality} closely. The only changes are that
\begin{align*}
&\mathrm{Pr}\left( \lnot R_{(\mathcal{D})}\bigcap \lnot R_{(\mathcal{Q})}  \Big| \underbrace{R_{(\mathcal{S})}\bigcap R_{(\mathcal{G})} ; \theta^{(\mathcal{S})}=\theta_{\mathcal{A}}^{(\mathcal{S})}, \theta^{(\mathcal{G})}=\theta_{\mathcal{A}}^{(\mathcal{G})}}_{:=\bpsi}\right) \geq \\
&  1-\underbrace{\sum_{s=1}^S\mathrm{Pr}\left( R_{(\mathcal{D, S})}^{(s)} \Big| \bpsi\right)}_{\leq\frac{S-1}{S+G+D+Q}\alpha_{-}} - \underbrace{\sum_{g=1}^G \mathrm{Pr}\left( R_{(\mathcal{D, G})}^{(g)} \Big| \bpsi\right)}_{=0} 
   - \underbrace{\sum_{t=1}^D\mathrm{Pr}\left( R_{(\mathcal{D})}^{(t)} \Big| \bpsi\right)}_{\leq\frac{D}{S+G+D+Q}\alpha_{-}}-
   \underbrace{\sum_{q=1}^Q\mathrm{Pr}\left( R_{(\mathcal{Q})}^{(q)} \Big| \bpsi\right)}_{\leq\frac{Q}{S+G+D+Q}\alpha_{-}} &\geq 1- \frac{S+D+Q-1}{S+G+D+Q}\alpha_{-}.
\end{align*}
Were the results under each sum is motivated as follows. First, we note that conditioned on that $R_{(\mathcal{G})}$ has occurred, and the conditions of the proposition, it follows that
$\mathrm{Pr}\left( R_{(\mathcal{D, G})}^{(g)} \mid \bkappa\right) = 0 \,\, \forall\,\, g = 1,...,G$. We can bound $\mathrm{Pr}\left( R_{(\mathcal{D, S})}^{(s)} \mid \bkappa\right) < \frac{1}{S+G+D+Q}\alpha_{-}$ for the worst case scenario of $S-1$ non-significant superiority tests conditioned on $R_{(\mathcal{S})}$. For the remaining $D$ deterioration tests, they are assumed under their null for the global alternative which means $\mathrm{Pr}\left( R_{(\mathcal{D})}^{(t)} \mid \bkappa\right) = \frac{1}{S+G+D+Q}\alpha_{-}\,\, \forall \,\, t=1,...,D$. 
Plugging this into the the proof of Proposition \ref{prop:alpha_beta_under_correction_with_deterioration_and_quality} and following the same steps solving for $\beta^*$ immediately implies that under the correction given in the proposition, the type II error rate is bounded by $\beta$.  
\end{proof}

It can seem counter-intuitive that including deterioration tests makes the overall rejection rule conservative given that Lemmas \ref{lemma:fpr_and_power_success_with_deterioration} and \ref{lemma:fpr_and_power_non_inf_deterioration} show that the rejection region for the deterioration and superiority or non-inferiority tests, respectively, are non-overlapping. There are several aspects at play here. First, there are potentially $D$ additional deterioration tests that can be falsely significant under the alternative. Second, importantly, only one success metric has to be superior to reject the overall null hypothesis---there are $D+S-1$ "chances" of having a falsely significant deterioration test even under the alternative, and when the other parts of the decision rule are fulfilled.
Finally, for guardrail metrics, it is indeed true that under the proposition's condition that $\alpha_- \leq \beta$, the rejection regions of the deterioration test and the non-inferiority test do not overlap rendering the deterioration tests for guardrail metrics completely redundant. However, we still use deterioration tests because we cannot fully control $\beta$ for each metric in practice. The type II error rate for a given NIM is achieved by setting the sample size, and there is only one sample size per experiment. Consequently, all metrics but the metric with the largest variance will be overpowered in practice. If there is a big difference in variance between some of the metrics, it is not implausible that the actual implied type II error rate for some guardrail metrics (given the NIM, variance, and the sample size used for the experiment) is smaller than $\alpha_-$. In this case, it is possible that the non-inferiority test and the deterioration test are both simultaneously significant. For completeness, if guardrail metrics are allowed to have type II error rates smaller than $\alpha_-$, the deterioration adjustment parameter should be adjusted to  $\phi=\frac{S-1+G+D}{S+G+D}\alpha_{-}$. 

\subsection{Additional approximate improvements}
The $\beta$ correction in the updated Proposition bounds the type II error rate under a quite extreme worst-case scenario. This is because under the alternative, and given that at least one success metric has improved, the probability that any of the remaining success metrics is significantly inferior is very low. In fact, the probability is so low that it is usually negligible. We formalize this in Remark \ref{remark:power_dr2}.
\begin{remark}\label{remark:power_dr2}
 If in the updated proposition, the test statistics of the deterioration tests and superiority tests for the success metrics are normally distributed, and the correction of $\beta$ is replaced with
 \begin{align*}
            \beta^*=\begin{cases} 
            \frac{\beta - \phi^*}{(1 -\phi^*)(G+1)} &   S>0\\ 
            \frac{\beta - \phi^*}{(1 -\phi^*)G}&   S=0 \end{cases}, 
        \end{align*} 
where $\phi^* = \frac{D+Q}{S+G+D+Q}\alpha_{-}$, the type II error rate is $\lessapprox \beta$.
\end{remark}
\begin{proof}
We are following the second part of the proof for the updated Proposition. Now, since $\mathrm{Pr}\left( R_{(\mathcal{D, S})}^{(s)} \mid \bkappa\right) = Pr(Z_{(\mathcal{D, S})}^{(s)}<c_{\alpha^*_-})$ where $Z\sim N(\theta_\mathcal{A}^{(S)}, 1)$ it follows that
For $\mathrm{Pr}\left( R_{(\mathcal{D, S})}^{(s)} \mid \bkappa\right)\rightarrow 0$ as $\beta \rightarrow 0$. Importantly, even for $\beta<0.5$ it is the case that $\mathrm{Pr}\left( R_{(\mathcal{D, S})}^{(s)} \mid \bkappa\right)\approx 0$ and this probability decreases with $S$, as $S$ goes into the correction of both $\alpha^*$ and $\alpha^*_-$. 
Table \ref{tab:prob_of_deterioration_under_alternative} displays some examples of the probability of significant deterioration for any of $S-1$ success metrics that might not be superior under $\bkappa$, i.e. $\sum_{s=1}^{S-1}\mathrm{Pr}\left( R_{(\mathcal{D, S})}^{(s)} \mid \bkappa\right)=(S-1)\mathrm{Pr}\left( R_{(\mathcal{D, S})}^{(s)} \mid \bkappa\right)$. Even for the extreme case when $\alpha_{+}=0.1$ and $\alpha_-=\beta=0.5$, the probability that deterioration is rejected among any of the $S-1$ (since conditioned on $R_{(\mathcal{S})}$ at least one success metric must be significantly superior which by Lemma \ref{lemma:fpr_and_power_success_with_deterioration} implies non-significant deterioration test for that metric) is only $0.02$. For all other cases the probability is less than $0.01$ and rapidly declining.   
\begin{table}[ht!]
\centering
\begin{tabular}{clllll}
  \hline
 $\beta=\alpha_{-}$ & $S=2$ & $S=5$ & $S=10$ & $S=20$ & $S=30$ \\ 
  \hline
0.2 & 0.000082 & 0.000007 & 0.000001 & 0.000000 & 0.000000 \\ 
  0.3 & 0.000674 & 0.000072 & 0.000010 & 0.000001 & 0.000000 \\ 
  0.4 & 0.003074 & 0.000411 & 0.000067 & 0.000010 & 0.000003 \\ 
  0.5 & 0.010188 & 0.001704 & 0.000322 & 0.000055 & 0.000019 \\ 
   \hline
\end{tabular}
\caption{The probability of significant deterioration for any of $S-1$ success metrics that might not be superior under $\bkappa$. $\beta$ is varying $0.2, 0.3, 0.4 0.5$, and as a worst case scenario $\alpha_-=\beta$. $\alpha_{+}=0.1$.}\label{tab:prob_of_deterioration_under_alternative}
\end{table}
This means that
\begin{align*}
   &\mathrm{Pr}\left( \lnot R_{(\mathcal{D})}\cap \lnot R_{(\mathcal{D})}\lnot R_{(\mathcal{Q})} \mid \bpsi\right) = \\
& \gtrapprox 1 - \frac{D+Q}{S+G+D+Q}\alpha_{-},
\end{align*}
from which the remark follows immediately from the previous proof.
\end{proof}

\section{Examples of global false and true positive rates} \label{sec:app:examples_fpr_power_dr1}
Although Proposition \ref{prop:alpha_beta_under_correction} guarantees that the risks are bounded, Tables \ref{tab:casesFPR} and \ref{tab:casesPower} provide further insights into how tight these bounds are under a few specific scenarios. Table \ref{tab:casesFPR} displays the global type I error rate with and without the corrections of Decision Rule \ref{dr1} under independence or perfect linear dependence. We assume that without correction $\alpha$ is used for all tests. 
\begin{table}[!hbt]
    \centering
    \begin{tabular}{cc|lll}
    \toprule
        \multicolumn{2}{c}{Metrics used} & & \multicolumn{2}{c}{Global type 1 error rate}\\
        \cline{1-2} \cline{4-5}
        Guardrail & \multicolumn{1}{c}{Success} & Covariance & Without correction & With correction \\
        \midrule
        \checkmark & x & Perfectly correlated & $\alpha$ & $\alpha$ \\
        \checkmark & x &Independent & $\alpha^G$ &$\alpha^G$ \\
        x & \checkmark& Perfectly correlated & $\alpha$ & $\frac{\alpha}{S}$ \\
        x & \checkmark &Independent & $1-(1-\alpha)^S$ & $1-(1-\frac{\alpha}{S})^S \lesssim \alpha$ \\
        \checkmark& \checkmark&All independent & $\alpha^{G}\times \left(1-(1-\alpha)^S\right)$ & $\alpha^{G}\times \left(1-(1-\frac{\alpha}{S})^S\right) \lesssim \alpha^{G+1}$ \\
        \checkmark& \checkmark&All perfectly correlated & $\alpha$ & $\frac{\alpha}{S}$\\
         \bottomrule
    \end{tabular}
    \caption{Type I error rate for Decision Rule \ref{dr1} under different settings of success and guardrail metrics with different covariance structures for the test statistics, and, the type I error rate in terms of the intended overall type I error rate under the applied correction. All cases are under the hypotheses null for all tests. The notation $a\lesssim b$ means $a<b+\epsilon$ where $\epsilon>0$.}\label{tab:casesFPR}
\end{table}
With only guardrail metrics, the global type I error rate ranges between $\alpha^G$ and $\alpha$, which follows from the nature of the IU test. When there are only success metrics involved, the global type I error rate is between $\alpha$ and $1-(1-\alpha)^S$, which makes the need for a proper adjustment immediately clear---with $\alpha=0.05$ and e.g. $S=5$ success metrics, the global type I error rate without a correction exceeds 20\%. When both types of metrics are included, the global type I error rate is always controlled and on the conservative side. The actual rate is determined by the number of each type of metric and on the dependency structure. 
Table \ref{tab:casesPower} presents the global power rates with and without the corrections of Decision Rule \ref{dr1} under independence or perfect linear dependence. We assume that without correction all tests are design to achieve $\beta$. 
\begin{table}[!hbt]
    \centering
    \begin{tabular}{cc|lll}
    \toprule
        \multicolumn{2}{c}{Metrics used} & & \multicolumn{2}{c}{Global true positive rate}\\
        \cline{1-2} \cline{4-5}
        Guardrail & \multicolumn{1}{c}{Success} & Covariance & Without correction & With correction \\
        \midrule
        \checkmark & x & Perfect corr. & $1-\beta$ & $1-\frac{\beta}{G}$ \\
        \checkmark & x & Independent & $(1-\beta)^G<1-\beta $ &$(1-\frac{\beta}{G})^G \gtrsim 1-\beta$ \\
        x & \checkmark& Perfect corr. & $1-\beta$ & $1-\beta$ \\
        x & \checkmark &Independent & $1-\beta^S$ & $1-\beta^S$ \\
        \checkmark& \checkmark&Independent & $(1-\beta)^G (1-\beta^S) < 1-\beta \forall G>1 $ & $ (1-\frac{\beta}{G+1})^G \left(1-\left(\frac{\beta}{G+1}\right)^S\right) > 1-\beta $ \\
        \checkmark& \checkmark&Perfect corr. & $1-\beta$ & $1-\frac{\beta}{G+1}$\\
         \bottomrule
    \end{tabular}
    \caption{True positive rate (power) for Decision Rule \ref{dr1} under different settings of success and guardrail metrics with different covariance structures for the test statistics, and, the true positive rate in terms of the intended overall true positive rate under the applied correction. All cases are under the hypotheses null for all tests. The notation $a\gtrsim b$ means $a>b+\epsilon$ where $\epsilon>0$.}\label{tab:casesPower}
\end{table}

\section{A note on sequential testing of deterioration}\label{app:seq_test}
In the previous sections it is shown that if a fixed horizon test is used, the type I error rate of the non-inferiority and superiority tests are not affected by adding deterioration tests under mild conditions on $\beta$ and $\alpha_{-}$. However, deterioration tests are typically added to ensure that the experiment can be aborted if some metrics start to deteriorate. To be able to maintain the type I error rate for the deterioration tests ($\alpha_{-}$), this requires using sequential tests for the deterioration tests. 
Sequential deterioration tests add complexity to the error rate calculations. This follows because even under conditions where the rejection regions for the deterioration and the superiority/non-inferiority tests are not overlapping at the last time point, it is possible that the deterioration test is significant at a previous intermittent analysis even in the case where the superiority or non-inferiority test is significant when the experiment is concluded. 

Let $K$ be the number of intermittent analyses for the deterioration test, such that the superiority tests and non-inferiority tests are only performed at $k=K$. The  rejection rates under some hypothesis is then given by
\begin{equation}
  Pr(R_\mathcal{S} \bigcap \left(\lnot R^{(1)}_\mathcal{-S} \cap \dots \cap \lnot R^{(K-1)}_\mathcal{-S}\right) \mid \theta)  
\end{equation}
 for the superiority test, 
and 
\begin{equation}
  Pr(R_\mathcal{G} \bigcap \left(\lnot R^{(1)}_\mathcal{-G} \cap \dots \cap \lnot R^{(K-1)}_\mathcal{-G} \right) \mid \theta)  
\end{equation}
for the non-inferiority test. 

It is difficult to bound these risks without referring to a specific test. E.g., in the Group Sequential Test (GST) literature, stopping if a treatment effect is going in the wrong direction is possible using so called futility bounds. It is well known that the incorporation of such bounds decrease the type I error rate and power for the superiority test. If one commits to always stopping if futility bounds are crossed, so-called "binding futility bounds", one can adjust the bounds for the superiority test to achieve the intended type I error rate. 
As we will show below, in practice, adding sequential deterioration tests will only mildly affect the rates. 

Here, we simply estimate the type I error rate and power of the superiority test and non-inferiority test with deterioration, respectively, by Monte Carlo simulation. One hundred intermittent analyses are performed for the deterioration test using GST\footnote{The simulation code is available in the online supplementary material, and it is easy to to exchange the GST for any other sequential test.}. The alternatives and NIM are selected to yield 80\% power for the fixed-horizon tests without deterioration. One million replications were performed for each setting.  

Table \ref{tab: MC_sup_and_noninf_with_deterioration} displays the results. 
\begin{table}[hbt!]
\centering
\begin{tabular}{rrrcllll}
\# Analyses &Metric & Hyp. & Design & sig. decision & sig. deteriorating  & sig. super. & sig. both directions \\ 
  \hline
100&success & H0 & $\alpha=.05$ & 0.04956 & 0.04974 & 0.04961 & 0.00005 \\ 
100&  success & H1 & $\alpha=.05,\beta=.2$ & 0.80077 & 0.00119 & 0.80135 & 0.00058 \\ 
100&  guardrail & H0 & $\alpha=.05,\beta=.2$ & 0.04792 & 0.72853 & 0.04901 & 0.00109 \\ 
100&  guardrail & H1 & $\alpha=.05,\beta=.2$ & 0.79278 & 0.04939 & 0.80069 & 0.00791 \\ 
100&  guardrail & H0 & $\alpha=.05,\beta=.05$ & 0.04399 & 0.91578 & 0.04901 & 0.00502 \\ 
100&  guardrail & H1 & $\alpha=.05,\beta=.05$ & 0.92964 & 0.04939 & 0.94997 & 0.02033 \\ 
   \hline
10&success & H0 & $\alpha=.05$ & 0.04982 & 0.05011 & 0.04985 & 0.00002 \\ 
 10& success & H1 & $\alpha=.05,\beta=.2$ & 0.79924 & 0.00060 & 0.79939 & 0.00015 \\ 
10&  guardrail & H0 & $\alpha=.05,\beta=.2$ & 0.04883 & 0.73604 & 0.04985 & 0.00101 \\ 
10&  guardrail & H1 & $\alpha=.05,\beta=.2$ & 0.79263 & 0.05011 & 0.79939 & 0.00676 \\ 
 10& guardrail & H0 & $\alpha=.05,\beta=.05$ & 0.04458 & 0.92031 & 0.04985 & 0.00526 \\ 
 10& guardrail & H1 & $\alpha=.05,\beta=.05$ & 0.93043 & 0.05011 & 0.94983 & 0.01940 \\ 
   \hline
\end{tabular}
\caption{False and true positive rates for the deterioration tests, the superiority test and non-inferiority test under the null and alternatives of the superiority and non-inferiority test. The "sig. decision" column displays the proportion of replications that had no significant deterioration and rejected the null hypothesis of the superiority or non-inferiority test, respectively. }\label{tab: MC_sup_and_noninf_with_deterioration}
\end{table}
For success metrics, the results indicate that using deterioration tests in addition to the superiority test for success metrics imposes negligible conservativeness under the null and the alternative.

Under the non-inferiority null hypothesis, the type I error rate decreases slightly due to deterioration, but to a degree without much practical relevance. The power of the deterioration test is 72\%, which is lower than 80\% because the deterioration test uses a GST test with a 100 intermittent analyses which decreases the overall power as compared to a fixed horizon test. Under the non-inferiority alternative, the rejection rate of the deterioration test decreases to the expected around 5\%, and the power of the non-inferiority is not practically affected by the deterioration test. 

It is of course natural to consider sequential tests also for quality tests, such as sample ratio mismatch. However, since these tests tends to separate from the success and guardrail metrics, using sequential tests has no implications. As long as the quality test that are used are valid, i.e., bounds the type I error rates, the propositions apply also with sequential tests for the quality metrics. In the next section we formalize the inclusion of quality metrics in the decision rule.

\section{Using Nyholt's method of efficient number of independent tests}\label{sec:app:nyholt}
As discussed above, the lack of assumption regarding the data generating process implies that under many processes the decision rule with have too low type I error rate and to high power. Although there are many alternatives to Bonferroni-type corrections we have two conditions that we would like to have for any design. First, the individual metric results should have CI's that maps to the decision rule outcome.
Second, the design should allow for analytical required sample size calculations. 
These conditions means that most alternative multiple correction methods, such as \cite{holm1979} and \cite{hommel1988} are discarded. Very few methods have confidence intervals, and even fewer are possible to integrate in the sample size calculation, due to their reliance on the p-values.  

One method that fulfills our conditions is Nyholt's method for calculating the so-called effective number of independent tests in a set of tests with arbitrary covariance structure \citep{nyholt2004}. This method is frequently applied in high-dimensional genome testing, and has been refined by several, see e.g. \cite{li2005adjusting, galwey2009new}. Nyholt's method is simple, the effective number of independent tests is given by 
\begin{equation}
M_{E} = 1 + (M-1)\left( 1-\frac{V_\lambda}{M} \right),
\end{equation}
where $M$ is the total number of tests and $V_\lambda$ is the sample variance of the eigenvalues of the covariance matrix.
We can incorporate the efficient number of independent tests in our design to be less conservative. Let $M_{E}(\boldsymbol{\Sigma})$ be the effective number of independent tests for a covariance matrix $\boldsymbol{\Sigma}$. Then we can simply update $\alpha^*=\frac{\alpha}{M_{E}(\boldsymbol{\Sigma}_\mathcal{S})}$ and similarly 
\begin{align*}
    \beta^*=\begin{cases} 
    \frac{\beta - \phi}{(1 -\phi)(M_{E}(\boldsymbol{\Sigma}_\mathcal{G})+1)} &   S>0\\ 
    \frac{\beta - \phi}{(1 -\phi)M_{E}(\boldsymbol{\Sigma}_\mathcal{G})}&   S=0. \end{cases}
\end{align*} 

\subsubsection{Monte Carlo simulation results with the Nyholt method}
Below the tables from Section \ref{sec:MC} are repeated together with the result using Nyholts method as described above. As expected, the method improves the efficiency. This is shown by type I error rates and power closer to the intended rates under certain covariance structures. 
\begin{table*}[t]
\centering
\begin{tabular}{rrccccc}
Covariance & Correction & $R_\mathcal{S}$ & $R_\mathcal{G}$ & $R_\mathcal{D,S}\bigcup R_\mathcal{D,G}$ & $R_\mathcal{D}\bigcup R_\mathcal{Q}$ & D. Rule 2\\ 
   \hline
Independent & None & 0.227 & 0.000 & 0.999 & 0.186 & 0.000 \\ 
  Independent & Only Alpha & 0.049 & 0.000 & 0.866 & 0.014 & 0.000 \\ 
  Independent & Prop. 4.1 & 0.049 & 0.000 & 0.999 & 0.014 & 0.000 \\ 
  Independent & Prop. 4.1+Nyholt & 0.049 & 0.000 & 0.999 & 0.014 & 0.000 \\ 
  Dependent & None & 0.063 & 0.040 & 0.773 & 0.186 & 0.032 \\ 
  Dependent & Only Alpha & 0.014 & 0.040 & 0.376 & 0.014 & 0.014 \\ 
  Dependent & Prop. 4.1 & 0.014 & 0.040 & 0.771 & 0.014 & 0.014 \\ 
  Dependent & Prop. 4.1+Nyholt & 0.059 & 0.040 & 0.692 & 0.033 & 0.038 \\ 
  Block 1 & None & 0.227 & 0.040 & 0.825 & 0.186 & 0.006 \\ 
  Block 1 & Only Alpha & 0.050 & 0.040 & 0.388 & 0.014 & 0.002 \\ 
  Block 1 & Prop. 4.1 & 0.050 & 0.040 & 0.776 & 0.014 & 0.002 \\ 
  Block 1 & Prop. 4.1+Nyholt & 0.050 & 0.040 & 0.625 & 0.020 & 0.002 \\ 
  Block 2 & None & 0.063 & 0.000 & 0.999 & 0.186 & 0.000 \\ 
  Block 2 & Only Alpha & 0.014 & 0.000 & 0.865 & 0.014 & 0.000 \\ 
  Block 2 & Prop. 4.1 & 0.014 & 0.000 & 0.999 & 0.014 & 0.000 \\ 
  Block 2 & Prop. 4.1+Nyholt & 0.059 & 0.000 & 0.999 & 0.020 & 0.000 \\ 
   \hline
\end{tabular}
\caption{Simulation results of the type I error rates under the null hypotheses of the non-inferiority and the superiority tests, respectively. The rejection of the deterioration, non-inferiority, and superiority tests are presented along with the rejection rate of Decision Rule \ref{dr2}. }
\end{table*}

\begin{table*}[t]
\centering
\begin{tabular}{llccccc}
  \hline
Covariance & Correction & $R_\mathcal{S}$ & $R_\mathcal{G}$ & $R_\mathcal{D,S}\bigcup R_\mathcal{D,G}$ & $R_\mathcal{D}\bigcup R_\mathcal{Q}$ & D. Rule 2 \\ 
  \hline
Independent & None & 0.227 & 0.328 & 0.399 & 0.184 & 0.047 \\ 
  Independent & Only Alpha & 0.050 & 0.328 & 0.035 & 0.015 & 0.016 \\ 
  Independent & Prop. 4.1 & 0.050 & 0.877 & 0.035 & 0.015 & 0.042 \\ 
  Independent & Prop. 4.1+Nyholt & 0.050 & 0.877 & 0.035 & 0.015 & 0.042 \\ 
  Dependent & None & 0.062 & 0.765 & 0.068 & 0.184 & 0.051 \\ 
  Dependent & Only Alpha & 0.013 & 0.765 & 0.006 & 0.015 & 0.013 \\ 
  Dependent & Prop. 4.1 & 0.013 & 0.966 & 0.006 & 0.015 & 0.013 \\ 
  Dependent & Prop. 4.1+Nyholt & 0.058 & 0.898 & 0.012 & 0.032 & 0.056 \\ 
  Block 1 & None & 0.228 & 0.765 & 0.272 & 0.184 & 0.113 \\ 
  Block 1 & Only Alpha & 0.048 & 0.765 & 0.023 & 0.015 & 0.036 \\ 
  Block 1 & Prop. 4.1 & 0.048 & 0.966 & 0.023 & 0.015 & 0.045 \\ 
  Block 1 & Prop. 4.1+Nyholt & 0.048 & 0.892 & 0.031 & 0.020 & 0.042 \\ 
  Block 2 & None & 0.062 & 0.329 & 0.272 & 0.184 & 0.016 \\ 
  Block 2 & Only Alpha & 0.013 & 0.329 & 0.023 & 0.015 & 0.004 \\ 
  Block 2 & Prop. 4.1 & 0.013 & 0.878 & 0.023 & 0.015 & 0.011 \\ 
  Block 2 & Prop. 4.1+Nyholt & 0.058 & 0.859 & 0.031 & 0.020 & 0.048 \\ 
   \hline
\end{tabular}
\caption{Simulation results of the rejection rates under the alternative hypothesis of the non-inferiority and the null hypothesis of superiority tests. The rejection rates of the deterioration, non-inferiority, and superiority tests are presented along with the rejection rate of Decision Rule \ref{dr2}.}
\end{table*}

\begin{table*}[t]
\centering
\begin{tabular}{llccccccc}
  \hline
Covariance & Correction & $R_\mathcal{S}$ & $R_\mathcal{G}$ & $R_\mathcal{D,S}\bigcup R_\mathcal{D,G}$ & $R_\mathcal{D}\bigcup R_\mathcal{Q}$ & D. Rule 2 \\ 
  \hline
Independent & None & 1.000 & 0.328 & 0.229 & 0.184 & 0.255 \\ 
  Independent & Only Alpha & 1.000 & 0.328 & 0.018 & 0.014 & 0.323 \\ 
  Independent & Prop. 4.1 & 1.000 & 0.874 & 0.017 & 0.014 & 0.859 \\ 
  Independent & Prop. 4.1+Nyholt & 1.000 & 0.874 & 0.017 & 0.014 & 0.859 \\ 
  Dependent & None & 0.832 & 0.766 & 0.064 & 0.184 & 0.618 \\ 
  Dependent & Only Alpha & 0.832 & 0.766 & 0.005 & 0.014 & 0.756 \\ 
  Dependent & Prop. 4.1 & 0.980 & 0.966 & 0.005 & 0.014 & 0.952 \\ 
  Dependent & Prop. 4.1+Nyholt & 0.935 & 0.898 & 0.011 & 0.031 & 0.869 \\ 
  Block 1 & None & 1.000 & 0.766 & 0.067 & 0.184 & 0.616 \\ 
  Block 1 & Only Alpha & 1.000 & 0.766 & 0.005 & 0.014 & 0.756 \\ 
  Block 1 & Prop. 4.1 & 1.000 & 0.966 & 0.005 & 0.014 & 0.952 \\ 
  Block 1 & Prop. 4.1+Nyholt & 1.000 & 0.892 & 0.007 & 0.019 & 0.874 \\ 
  Block 2 & None & 0.832 & 0.329 & 0.227 & 0.184 & 0.213 \\ 
  Block 2 & Only Alpha & 0.832 & 0.329 & 0.017 & 0.014 & 0.270 \\ 
  Block 2 & Prop. 4.1 & 0.980 & 0.874 & 0.017 & 0.014 & 0.842 \\ 
  Block 2 & Prop. 5.1+Nyholt & 0.977 & 0.855 & 0.024 & 0.019 & 0.814 \\ 
   \hline
\end{tabular}
\caption{Simulation results of the type I error rate under the alternative hypotheses of the non-inferiority and the superiority tests, respectively. The rejection rates of the deterioration, non-inferiority, and superiority tests are presented along with the rejection rate of Decision Rule \ref{dr2}.}
\end{table*}

\end{appendix}



\begin{acks}[Acknowledgments]
The authors would like to thank Bob Wilson for feedback on earlier drafts of this paper.
\end{acks}

\pagebreak

\bibliographystyle{imsart-number} 
\bibliography{library}       

\end{document}